\documentclass[10pt]{article}
\usepackage[T1]{fontenc}
\usepackage{lmodern}
							
\usepackage[frenchmath]{mathastext}
\usepackage{natbib}       									
\usepackage{amsmath}                        
\usepackage{graphicx} 										
\usepackage{caption}
\usepackage{booktabs}
\usepackage{subfigure}
\usepackage{float} 	
\usepackage{enumitem}												
\usepackage{mdwlist}												
\usepackage[dvipsnames]{xcolor}							
\usepackage[plainpages=false, pdfpagelabels]{hyperref} 
	\hypersetup{
		colorlinks   = true,
		citecolor    = RoyalBlue,
		linkcolor    = RubineRed, 
		urlcolor     = Turquoise
	}
\usepackage{amssymb}                        
\usepackage[mathscr]{eucal}                 
\usepackage{dsfont}									
\usepackage[top=1.0in, bottom=1.0in, left=0.8in, right=0.8in]{geometry} 
\usepackage{ulem}
\usepackage{mathtools}                      
\mathtoolsset{showonlyrefs=true}       
\usepackage{amsthm}                
	\allowdisplaybreaks                       
	\theoremstyle{plain}
	\newtheorem{theorem}{Theorem}
	\newtheorem{proposition}[theorem]{Proposition}
	
	\theoremstyle{definition}
	\newtheorem{definition}{Definition}

\newcommand\Eb{\mathds{E}}

\newcommand\Pb{\mathds{P}}
\newcommand\Qb{\mathds{Q}}
\newcommand\Rb{\mathds{R}}

\newcommand\Ac{\mathscr{A}}

\newcommand\Fc{\mathscr{F}}

\newcommand\Tc{\mathscr{T}}

\newcommand\del{\partial}
\newcommand\Del{\Delta}

\newcommand\Ebt{\widetilde{\Eb}}
\newcommand\Pbt{\widetilde{\Pb}}
\newcommand\Act{\widetilde{\Ac}}

\newcommand\Wt{\widetilde{W}}

\newcommand\Zt{\widetilde{Z}}

\newcommand{\dd}{\mathrm{d}}

\newcommand{\eqlnostar}[2]{\begin{align}\label{#1}#2\end{align}}
\newcommand{\eqstar}[1]{\begin{align*}#1\end{align*}}
\usepackage{bbold}
\makeatother

\begin{document}
\title{Sharing of longevity basis risk in pension schemes with income-drawdown guarantees}

\author{
Ankush Agarwal
\thanks{\textbf{e-mail}: \url{ankush.agarwal@glasgow.ac.uk}}
\qquad
Christian-Oliver Ewald
\thanks{\textbf{e-mail}: \url{christian.ewald@glasgow.ac.uk}}
\qquad 
Yongjie Wang
\thanks{\textbf{e-mail}: \url{yongjie.wang@glasgow.ac.uk}}
\\
Adam Smith Business School, University of Glasgow, G12 8QQ Glasgow, United Kingdom
}

\date{\today}

\maketitle
\begin{abstract}
This work studies a stochastic optimal control problem for a pension scheme which provides an income-drawdown policy to its members after their retirement. To manage the scheme efficiently, the manager and members agree to share the investment risk based on a pre-decided risk-sharing rule. The objective is to maximise both sides' utilities by controlling the manager's investment in risky assets and members' benefit withdrawals. We use stochastic affine class models to describe the force of mortality of the members' population and consider a longevity bond whose coupon payment is linked to a survival index. In our framework, we also investigate the longevity basis risk, which arises when the members' and the longevity bond's reference populations show different mortality behaviours. By applying the dynamic programming principle to solve the corresponding HJB equations, we derive optimal solutions for the single- and sub-population cases. Our numerical results show that by sharing the risk, both manager and members increase their utility. Moreover, even in the presence of longevity basis risk, we demonstrate that the longevity bond acts as an effective hedging instrument.


\end{abstract}

\section{Introduction }
Pension schemes face a wide range of risks such as investment risk, interest rate risk, inflation risk and longevity risk. Among all the risks, the longevity risk is becoming a non-negligible challenge to pension schemes. The longevity risk is the risk that the actual life expectancy may be higher than anticipated. It is a positive trend for society that people's average life expectancy is increasing over the last decades. However, the pension schemes are suffering from the loss caused by the unexpected increasing benefit outgo accompanied by the longevity trend. \cite{cocco2012longevity} shows that the average life expectancy of 65-year-old US (UK) males increases by 1.2 (1.5) years per decade. As a consequence, a defined benefit scheme (DB scheme) for those populations would have needed 29\% more wealth in 2007 than in 1970. 

Instead of a classical DB or DC scheme, this work considers a pension scheme which provides an income-drawdown option to its members in the decumulation phase. The decumulation phase refers to the period after the retirement of pension scheme members. The drawdown option gives the members the right to withdraw money periodically from the scheme until death time, while keeps the remaining amount in the pension scheme. The management of the scheme starts from the retirement time of the members and ends until the last member passes away. Thus, the investment period can last for decades. Due to the long investment horizon, the scheme may face significant interest rate risk and inflation risk. However, the main focus of this work is on hedging the longevity risk. We assume the risk-free interest rate to be constant and do not consider the inflation risk. In other words, we only consider the investment risk and longevity risk. The scheme manager invests in a representative stock to gain investment returns, and he is exposed to the investment risk. As the members tend to live longer, the number of surviving members each year is more than expected, resulting in an increase of benefit withdrawals. Also, the length of decumulation period is increasing due to the members' growing life expectancy. One way to reduce the pension scheme's unexpected loss caused by its exposure to the longevity risk is to gain a better understanding of the future mortality evolution. Another way is to seek some financial instruments to transfer the longevity risk to the financial market.

The force of mortality, which represents the instantaneous rate of mortality at a certain time, is useful when studying mortality behaviour. In a continuous-time framework, \cite{luciano2005non} modelled the force of mortality by affine processes and calibrated the models to the observed and projected UK mortality tables. They claimed that the affine process with a deterministic part that increases exponentially could describe the force of mortality properly. This paper applies Ornstein-Uhlenbeck (OU) and Cox-Ingersoll-Ross (CIR) process to model the force of mortality of certain populations. To hedge the longevity risk using financial instruments, \cite{blake2001survivor} introduced a survivor bond which provides coupon payment based on the number of survivors in a chosen reference population. Under a stochastic force of mortality framework, \cite{menoncin2008role} studied the optimal consumption and investment strategy for an individual investor using the longevity bond to hedge the investor's longevity risk. He argued that the longevity bond plays a crucial role in the individual's longevity risk management.  However, the longevity basis risk arises when the force of mortality of the longevity bond reference population correlates imperfectly with the pension members' force of mortality. To deal with the longevity basis risk, \cite{wong2014time} applied the cointegration technique to study the mean-variance hedging of longevity risk for an insurance company with a longevity bond. They suggested that cointegration is vital in longevity risk management. In this paper, we introduce a rolling zero-coupon longevity bond to hedge the scheme's longevity exposure.  We show that the longevity bond provides an efficient way to hedge the longevity risk both in the cases with and without longevity basis risk.

Under a continuous-time framework, this work aims to determine the optimal investment and benefit withdrawal strategy for a pension scheme that provides an income-drawdown policy in the decumulation phase. After retirement, the members continuously withdraw money from the pension scheme until death. Upon death, a deterministic proportion of the passed away members' pension account balance is given to the manager as compensation. The proposed optimal control problem considers both benefit of the scheme members and profit of the manager by incorporating a risk-sharing rule parameter. By applying the dynamic programming principle, we first solve the optimisation problem under a general case. Then, we drive explicit solutions for single- and sub-population cases. Numerical studies on single-population and sub-population OU models are conducted to investigate the optimal portfolio strategy and optimal benefit withdrawal rate dynamically. Comparison studies are provided to analyse the performance of longevity risk hedge. We also implement sensitivity analyses to look into the impact of the market price of risk and risk-sharing rule parameter. Our results show that with the absence of longevity basis risk, the longevity bond provides efficient longevity risk hedge when the members show a longevity trend. When the members are a sub-population of the longevity bond reference population, the presence of the longevity basis risk may weaken the longevity bond's hedging performance. However, it still provides a way to hedge the longevity risk and offers a risk premium. We also establish that an equal-risk-sharing rule is beneficial to both the members and the manager in the long run. This paper contributes to the literature by studying a stochastic optimal control problem for pension schemes in the presence of longevity risk and longevity basis risk. Our optimal control problem also incorporates a risk-sharing rule parameter which decides the agreement on how to share the risk between the members and the manager such that both parties benefit in the long run.

The rest of this paper is organised as follows. Section \ref{sec:model} introduces the mathematical framework of the problem and derives the optimal solution in the general case. Explicit solutions under single-population and sub-population models are given in Section \ref{sec:solution}. In Section \ref{sec:numerical}, numerical simulations are carried out to investigate the optimal portfolio strategy and benefit withdrawal. Comparison and sensitivity analyses are also conducted to discuss the role of longevity bond and impact of model parameters in the optimal solutions. Section \ref{sec:conclusion} concludes the paper.

\section{Mathematical Framework}
\label{sec:model}
The literature on optimal control problems for pension schemes with deterministic mortality behaviour of the populations is rich. However, given the fluctuations in the mortality behaviour over time, it is more practical to use stochastic mortality rates. In this work, we use affine class models for the force of mortalities of members' population and the reference population for longevity bond. When, the two populations are different, we use $n=2$ populations in our framework and when they are the same, we use $n=1$ population.

\subsection{The stochastic force of mortality}
We consider an infinite time horizon $\Tc=[0,\infty)$ where time 0 represents the retirement time of all the populations. Let $\{{W(t)\mid t\in \Tc\}}=\{(W_1 (t), ..., W_n (t),W_S (t))^\prime \mid t\in \Tc\}$ denote an $(n+1)$-dimensional standard Brownian motion on a complete filtered probability space $(\Omega, \Fc, \{\Fc(t)\}_{t\geq0}, \Pb)$. Here, $n$ denotes the number of populations and $\Pb$ denotes the physical measure where we observe the longevity behaviours of the populations, and the financial market.

For $i \leq n,$ let $p_i(t)$ and $\lambda_i(t)$ denote the survival probability and the force of mortality, respectively, of the $i$th population at time $t$. The two are related through the following relation:
\eqstar{
\frac{\dd p_i(t)}{p_i(t)}= -\lambda_i(t)\dd t,& &p_i(0) = 1.
}
For any $s \geq t \in \Tc$, $\frac{p_i(s)}{p_i(t)}=e^{-\int_t^s\lambda_i(u)\dd u}$ is the survival probability of the $i$th population between time $t$ and $s$. For notational simplicity, we denote by $\{\lambda(t)\mid t\in \Tc\}=\{(\lambda_1(t),...,\lambda_n(t))^\prime\mid t\in\Tc\}$, the vector of force of mortalities and assume that it evolves as
\eqlnostar{eq:lambdas}{
\dd \lambda(t)= \mathscr{A} (t,\lambda)\dd t+\Sigma^\prime (t,\lambda)\dd W(t),
}
where 
\eqstar{
\mathscr{A}(t,\lambda)=\left[\begin{array}{ccc}
    a_1(t,\lambda)\\
    \vdots  \\
    a_{n}(t,\lambda)
    \end{array} \right], &&
    \Sigma^\prime (t,\lambda)= \left[\begin{array}{ccccc}
    \sigma_1(t,\lambda_1) &\cdots & \sigma_{1,n}(t,\lambda_{n})&0  \\
    \vdots &\ddots &\vdots &\vdots \\
    \sigma_{n,1}(t,\lambda_1) & \cdots & \sigma_{n,n}(t,\lambda_{n}) & 0 
    \end{array} \right].
}
For any $i,j= 1,\dots,n$, $a_i(t,\lambda)$ and $\sigma_{ij}(t,\lambda_j)$ are assumed to be continuous functions. In the following, we will use different affine class models for $\lambda(t).$
\subsection{The financial market} 
We consider a frictionless financial market consisting of a stock and a \textit{rolling zero-coupon longevity bond}. The money market account is denoted by $R(t)$,
\eqstar{
\frac{\dd R(t)}{R(t)} = r\dd t, & & R(0)=1, 
}
where $r$ denotes the constant risk-free interest rate. In this work we consider a constant risk-free rate of interest as our focus is to understand the impact of longevity and investment risk in a pensions scheme. Our analysis can also be performed in the presence of a stochastic interest rate.

Under a risk-neutral probability measure, the price of a financial derivative is the discounted expected value of its future payoff. We thus introduce an equivalent risk-neutral probability measure $\Qb$ by the following Radon-Nikodym derivative
\eqstar{
\frac{\dd \Qb}{\dd \Pb} = Z(T)=\exp\left(-\int_0^T\theta^\prime(t,\lambda)\dd W(t)-\frac{1}{2}\int_0^T\|\theta(t,\lambda)\|^2\dd t\right),
}
where $\{\theta(t,\lambda)\mid t\in \Tc\}=\{(\theta_1(t,\lambda_1),\cdots,\theta_n(t,\lambda_n),\theta_S)^\prime\mid t\in \Tc\}$ is $\Rb^n$-valued, $\Fc$-adapted process such that $Z(t)$ is a martingale and $\Eb[Z]=1$. $\theta(t,\lambda) $ is called the vector of market prices of risks and measures the additional amount of investment return when risk increases by one unit. By Girsanov's theorem, $\{W^\Qb(t)\mid t\in \Tc\}$ is an $(n+1)$-dimensional standard Brownian motion under $\Qb$ such that
\eqlnostar{Girsanov}{
W^\Qb(t)=W(t)+\int_0^t\theta(s,\lambda)\dd s.
}
The stock price process $\{S(t)\mid t \in \Tc\}$ is given as
\eqstar{
\frac{\dd S(t)}{S(t)} = r\dd t+\sigma_S\left[\theta_S\dd t+\dd W_S(t)\right],& &S(0) = S_0,
}
where $\sigma_s$ denotes the constant stock price volatility. The risk premium of the stock is $\theta_S\sigma_S$. 

In the literature, several types of mortality-linked securities are proposed to hedge the longevity risk. The values of these securities depend on the mortality index for some given populations: higher the survival rate, more valuable the securities. The definition below provides a short description of one such security, a \textit{zero-coupon longevity bond}.
\begin{definition}
A zero-coupon longevity bond is a contract paying a face amount which is equal to the survival probability of the reference population from time 0 until a fixed maturity time $T$.
\end{definition}
There may be multiple longevity bonds based on different reference populations in the market. However, as an illustration of the use of longevity bond, we only consider one longevity bond in this work. Let the $i$th population be the reference population of the zero-coupon longevity bond which pays $\frac{p_i(T)}{p_i(0)}$ at maturity. The arbitrage-free price of the longevity bond at time $t$ is given as
\eqstar{
L(t,T) = \Eb^\Qb_t\left[\frac{R(T)}{R(t)}\frac{p_i(T)}{p_i(0)}\right] =\Eb^\Qb_t\left[e^{-r(T-t)-\int_0^T\lambda_i(u)\dd u}\right] =e^{-\int_0^t\lambda_i(u)\dd u-r(T-t)}\Eb_t^\Qb\left[e^{-\int_t^T\lambda_i(u)\dd u}\right].
}
Applying It\^{o}'s formula and using \eqref{Girsanov} gives
\eqstar{
\frac{\dd L(t,T)}{L(t,T)} = r\dd t+\sum_{i,j=1}^{n}\sigma_L^{ij}(t,T) \Big[\theta_j(t,\lambda_j)\dd t+\dd W_j(t)\Big],
}
where $\sigma_L^{ij}(t,T)=\frac{\del L(t,T)}{\del\lambda_j(t)}\frac{\sigma_{ij}(t,\lambda_j)}{L(t,T)}$.

For any $t \in [0,T]$, the time to maturity of the zero-coupon longevity bond is $T-t$. In practice, it is impossible for an investor to find all zero-coupon longevity bonds in the market with any time to maturity $T-t$. By taking inspiration from the arguments in \cite{boulier2001optimal} on rolling zero-coupon bonds, we introduce a \textit{rolling zero-coupon longevity bond} $L(t)$ (with a little abuse of notation) with constant time to maturity $T$. The use of a rolling zero-coupon longevity bond in our set-up simplifies the calculations. The dynamics of $L(t)$ under $\Pb$ is given as
\eqstar{
\frac{\dd L(t)}{L(t)} = r\dd t+\sum_{i,j=1}^{n}\sigma_L^{ij}(t)\left[\theta_{j}(t,\lambda_j)+\dd W_j(t)\right].
}
From the above, we can see that $L(t)$ provides a longevity risk premium of $\sum_{i,j=1}^{n}\sigma_L^{ij}(t)\theta_{j}(t,\lambda_j)$. Any zero-coupon longevity bond $L(t,T)$ can be replicated by using the rolling longevity bond $L(t)$ and cash. The following equation shows the relationship between $L(t,T)$ and $L(t)$
\eqstar{
\frac{\dd L(t,s)}{L(t,s)} = \left(1-\frac{\sum_{i,j=1}^{n}\sigma_L^{ij}(t,s)}{\sum_{i,j=1}^{n}\sigma_L^{ij}(t)}\right)\frac{\dd R(t)}{R(t)}+\frac{\sum_{i,j=1}^{n-1}\sigma_L^{ij}(t,s)}{\sum_{i,j=1}^{n}\sigma_L^{ij}(t)}\frac{\dd L(t)}{L(t)}.
}

\subsection{The optimisation problem}

We study the optimal benefit withdrawal rate and investment strategy in the decumulation phase for a pension scheme that provides an income-drawdown option. That is, the members are allowed to withdraw money periodically from the scheme until death. Upon death, $\pi\in [0,1]$ fraction of the passed away members' pension balance is delivered to the scheme manager as compensation. While the remaining $1-\pi$ fraction of the balance stays in the scheme's fund pool. Let $\alpha_S(t)$, $\alpha_L(t)$ and $\alpha_0(t)$ denote the investments in stock, rolling longevity bond and money market, respectively. In addition, let $\beta(t)$ denote the amount of benefit withdrawal and $Y(t)$ the scheme's wealth. We assume that the force of mortality of pension members is described by $\lambda_j(t)$. In other words, population $j$ represents the pension members.

To study the dynamics of the scheme's wealth, we employ a similar method used in \cite{he2013optimalb} by first looking at the discrete-time changes in scheme's wealth. For any $t\in \Tc$ and a small positive number $\Del$, the returns on stock, longevity bond and money market account are given by
\eqstar{
&\frac{S(t+\Del)-S(t)}{S(t)}\alpha_S(t),& \frac{L(t+\Del)-L(t)}{L(t)}\alpha_L(t), & & \frac{R(t+\Del)-R(t)}{R(t)}\alpha_0(t).
}
Let $\mu(t,t+\Del)$ denote the rate of investment return. Then, we have
\eqstar{
\mu(t,t+\Del)Y(t) =& \frac{S(t+\Del)-S(t)}{S(t)}\alpha_S(t)+\frac{L(t+\Del)-L(t)}{L(t)}\alpha_L(t)+\frac{R(t+\Del)-R(t)}{R(t)}\alpha_0(t).
}
Let $q(t,t+\Del)$ denote the proportion of the members who pass away in time interval $(t,t+\Del).$ The total amount of passed away members' pension equals to $q(t,t+\Del)Y(t)$. During this time period, the scheme manager receives $\pi q(t,t+\Del)Y(t)$ as compensation. In the meantime, the members withdraw $\beta(t)\Del$ as benefits. Thus, the scheme's cashflow in this period is given by
\eqstar{
\mu(t,t+\Del)Y(t)-\pi q(t,t+\Del) Y(t)-\beta(t)\Del.
}
At time $t+\Del$, $1-q(t,t+\Del)$ fraction of the members survive, and the scheme wealth is equally distributed into each surviving member's pension account. Therefore, we have
\eqstar{
Y(t+\Del)= &\Big(Y(t)+\mu(t,t+\Del) Y(t)-q(t,t+\Del)\pi Y(t)-\beta(t)\Del\Big)\frac{1}{1-q(t,t+\Del)}.
}
Since $\lambda_j(t)$ is the force of mortality of members, we have $1-q(t,t+\Del) = e^{-\int_t^{t+\Del}\lambda_j(u)\dd u}$. The Taylor series approximation gives
\eqstar{
\frac{1}{1-q(t,t+\Del) } =e^{\int_t^{t+\Del}\lambda_j(u)\dd u} =1+\lambda_j(t)\Del+o(\Del).
}
Thereafter, we have
\eqstar{
Y(t+\Del)= &\Big(Y(t)+\mu(t,t+\Del) Y(t)-q(t,t+\Del)\pi Y(t)-\beta(t)\Del\Big)\Big(1+\lambda_j(t)\Del+o(\Del)\Big).
}
Since
\eqstar{
\mu(t,t+\Del)\lambda_j(t)\Del = o(\Del), & & q(t,t+\Del)\lambda_j(t)\Del=o(\Del),
}
we get
\eqstar{
Y(t+\Del)= Y(t)+\mu(t,t+\Del) Y(t)-q(t,t+\Del)\pi Y(t)-\beta(t)\Del+\lambda_j(t)\Del Y(t)+o(\Del).}
Furthermore, 
\eqstar{
\frac{Y(t+\Del)-Y(t)}{\Del} = \frac{\mu(t,t+\Del) }{\Del}Y(t)-\pi Y(t)\frac{q(t,t+\Del)}{\Del}-\beta(t)+\lambda_j(t)Y(t)+o(\Del).
}
The following equations hold:
\eqstar{
\lim\limits_{\Del \to 0}\frac{q(t,t+\Del)}{\Del}=&\lambda_j(t),\\ \lim\limits_{\Del \to 0}\mu(t,t+\Del)Y(t)=&\frac{\dd S(t)}{ S(t)}\alpha_S(t)+\frac{\dd L(t)}{L(t)}\alpha_L(t)+\frac{\dd R(t)}{R(t)}\alpha_0(t).
}
Thus, taking $\Del \to 0$ we obtain the continuous-time version of the scheme's wealth dynamics:
\eqstar{
\dd Y(t)=&\Big[ rY(t)+(1-\pi)\lambda_j(t)Y(t)+\alpha_S(t)\sigma_S\theta_S+\alpha_L(t)\sum_{i,j=1}^{n-1}\sigma_L^{ij}(t)\theta_j(t)-\beta(t)\Big]\dd t\\
&+\alpha_S(t)\sigma_S\dd W_S(t)+\alpha_L(t)\sum_{i,j=1}^{n-1}\sigma_L^{ij}(t)\dd W_j(t).
}
Here, we use the fact that $\alpha_0(t)=Y(t)-\alpha_S(t)-\alpha_L(t)$. It is worth noticing that the wealth process here is not a self-financing wealth process since there are continuous benefits withdrawal from the scheme. Also, at any time $t\in \Tc$, the manager receives $\pi \lambda_j(t)Y(t)\dd t$ amount of money out from the scheme as compensation. Meanwhile, the scheme continuously receives $(1-\pi)\lambda_j(t)Y(t)$ from the passed away members which is called the mortality credit in \cite{he2013optimala}.

At any time, the pension scheme manager decides the benefits withdrawal rate and the investment strategy. The manager works not only for his own benefit but also for the benefit of scheme members. This is also known as first-best principal-agent problem where the agent is paid a fraction of the scheme's wealth at the stochastic death time. More specifically, we consider an optimisation problem which combines the manager's and the members' utilities. Denote by $\tau$ the stochastic death time of the pension members, the objective function is given as
\eqstar{
J(Y,\lambda ; \alpha_S,\alpha_L, \beta)&= \Eb\Big[\int_0^\infty e^{-rs} U_P(\beta(s)) \mathbb{1}_{\{\tau\ge s\}}\dd s\Big]+\phi\Eb\Big[e^{-r\tau}U_A(\pi Y(\tau))\Big]\\
&=\Eb\Big[\int_0^\infty e^{-rs}p_j(s) \Big(U_P(\beta(s))+\phi \lambda_j(s)U_A(\pi Y(s))\Big)\dd s \Big]
}
where $U_P(\cdot)$ and $U_A(\cdot)$ denote the utility functions of the principal (members) and the agent (manager). The non-negative constant $\phi$ can be viewed as a parameter that determines the risk-sharing rule between the principal and the agent. The case $\phi=0$ corresponds to the situation when the manager works for only the sake of members. In this case, the objective is to maximise the members' running utility from benefit withdrawals while the manager pays no attention to his own utility. The case $0<\phi<1$ prioritizes the members' utility. When $\phi=1$, the objective function puts equal importance on members' and manager's utility.

To specify the optimisation problem, we set $U_P(\cdot)$ and $U_A(\cdot)$ as log utility functions:
\eqstar{
&U_A(x)=\ln x,&U_P(x)=\ln x,& &\forall x>0.
}
Besides, upon death, we assume the total amount of passed away members' remaining pension is paid to the manager (i.e. $\pi=1$). The wealth process is now
\eqlnostar{eq:wealth}{
\dd Y(t)=&\Big[ rY(t)+\alpha_S(t)\sigma_S\theta_S+\alpha_L(t)\sum_{i,j=1}^{n-1}\sigma_L^{ij}(t)\theta_j(t)-\beta(t)\Big]\dd t\nonumber\\
&+\alpha_S(t)\sigma_S\dd W_S(t)+\alpha_L(t)\sum_{i,j=1}^{n-1}\sigma_L^{ij}(t)\dd W_j(t).
}
It should be noticed that, at any time $t\in\Tc$, the manager receives $c(t)=\lambda_i(t)Y(t)$ as his compensation.

The optimisation problem is thus defined as
\eqlnostar{eq:problem}{
\left\{\begin{array}{cc}
     &\underset{\alpha_S,\alpha_L, \beta}{\sup} \, \Eb\Big[\int_0^\infty e^{-rs}p_j(s) \Big(\ln(\beta(s))+\phi \lambda_j(s)\ln( Y(s))\Big)\dd s \Big] \\
     &\text{s.t.\quad~ \eqref{eq:lambdas} and \eqref{eq:wealth} hold}.
\end{array}\right.
}
This optimal control problem can be solved by applying the dynamic programming principle, and the solution is given in the following proposition.
\begin{proposition}
The solution to optimisation problem \eqref{eq:problem} is
\eqstar{
&\frac{\beta^*(t)}{Y(t)}=\frac{1}{G(t,\lambda)}, \qquad \frac{\alpha_S^*(t)}{Y(t)}=\frac{\theta_S}{\sigma_S},\\
&\frac{\alpha_L^*(t)}{Y(t)} =\frac{\sum_{i,j=1}^{n-1}\sigma_L^{ij}(t)\theta_j (t,\lambda_j)}{\sum_{j=1}^{n-1}\left(\sum_{i=1}^{n-1}\sigma_L^{ij}(t)\right)^2 }+\frac{\sum_{i,j=1}^{n-1}\sigma_L^{ij}(t)\Sigma^\prime G_{\lambda}(t,\lambda)}{\sum_{j=1}^{n-1}\left(\sum_{i=1}^{n-1}\sigma_L^{ij}(t)\right)^2 G(t,\lambda)},
}
where
\eqstar{
G(t,\lambda) = \Eb_t\left[\int_t^\infty (\phi\lambda_j(s)+1)e^{-\int_t^s(r+\lambda_j(u))\dd u}\dd s\right].
}
\label{prop:single}
\end{proposition}
\begin{proof}
The proof is given in Appendix \ref{ap:single}.
\end{proof}
\section{Explicit solutions}
\label{sec:solution}
This section gives the explicit solutions to the optimal control problem proposed before in single-population and sub-population cases. In the literature, e.g. \cite{luciano2005non} and \cite{wong2014time}, several continuous-time stochastic models for force of mortality are studied. For example, OU process,  CIR process and Feller process. In this section, we use OU and CIR processes to model the stochastic force of mortality.
\subsection{Single-population model}
\label{sec:single}
Let $n=1$, that is, we assume that the reference population of the longevity bond happens to be the scheme members' population. We expect that the investment in the longevity bond can hedge the scheme's longevity risk effectively, since the uncertainty in the longevity bond value correlates the members' longevity risk perfectly. Now, \eqref{eq:lambdas} takes the following form:
\eqstar{
\dd \lambda_1(t)=\left(a_1(t)-b_1\lambda_1(t)\right) \dd t+\sigma_1(t,\lambda_1)\dd W_1 (t),
}
where $b_1$ is a constant. According to \cite{menoncin2017longevity}, $a_1(t)$ is defined as
\eqstar{
a_1(t) = b_1\left(\nu_1+\frac{1}{\Del_1}\left( 1+\frac{1}{b_1 \Del_1}\right)e^{\frac{t-m_1}{\Del_1}}\right),
}
where $m_1$, $\nu_1$ and $\Del_1$ are constants. Particularly, $m_1$ equals the modal value of life expectancy of the members. The initial value of the force of mortality is calculated according to the Gompertz-Makeham function
\eqstar{
\lambda_1(0) = \nu_1+\frac{1}{\Del_1}e^{-\frac{m_1}{\Del_1}}.
}

The continuous function $\sigma_1(t,\lambda_1)$ takes different forms for different processes. For instance: 
\begin{itemize}
\item{ OU process: $\sigma_1(t,\lambda_1)=\sigma_1$, and }
\eqlnostar{eq:singleOU}{
\dd \lambda_1(t)=\left(a_1(t)-b_1\lambda_1(t)\right) \dd t+\sigma_1\dd W_1 (t).
}
\item{ CIR process: $\sigma_1(t,\lambda_1)=\sigma_1\sqrt{\lambda_1(t)}$, and }
\eqlnostar{eq:singleCIR}{
\dd \lambda_1(t)=\left(a_1(t)-b_1\lambda_1(t)\right) \dd t+\sigma_1\sqrt{\lambda_1}\dd W_1 (t).
}
\end{itemize}

In what follows, we apply dynamic programming principle to solve the single-population optimisation problem. Proposition \ref{prop:singleOU} and \ref{prop:singleCIR} give the explicit solutions under OU process and CIR process settings, respectively.
\begin{proposition}
Suppose that $\lambda_1(t)$ follows the OU process \eqref{eq:singleOU}. Then, we have
\eqstar{
\Eb_t\left[e^{-\int_t^s\lambda_1(u)du}\right]=e^{A_0(t,s)-A_1(t,s)\lambda_1(t)},
}
where 
\eqstar{
&A_1(t,s) = \frac{1-e^{-b_1(s-t)}}{b_1},
& &A_0(t,s) = -\int_t^s\left(a_1(u)A_1(u,s)-\frac{1}{2}\sigma_1^2 A_1^2(u,s) \right)\dd u.
}
Suppose $\theta_1(t,\lambda_1):=\theta_1.$ Then, the solution to the single-population optimisation problem \eqref{eq:problem} is given as
\eqstar{
\frac{\beta^*(t)}{Y(t)}= \frac{1}{G(t,\lambda_1)}, && \frac{\alpha_S^*(t)}{Y(t)}= \frac{\theta_S}{\sigma_S},&& \frac{\alpha_L^*(t)}{Y(t)} = \frac{\theta_1}{\sigma_L(t)}+\frac{\sigma_1}{\sigma_L(t)}\frac{G_{\lambda_1}(t,\lambda_1)}{G(t,\lambda_1)},
}
where
\eqstar{
\sigma_L (t)= &-A_1(t,t+T)\sigma_1,\\
G(t,\lambda_1) = &\ \phi \lambda_1(t) \int_t^\infty e^{-(b_1+r)(s-t)+A_0(t,s)-A_1(t,s)\lambda_1(t)}\dd s\\
&+\phi\int_t^\infty\int_t^s \left[a_1(u)-\sigma_1^2A_1(u,s)\right]\dd u \ e^{-r(s-t)+A_0(t,s)-A_1(t,s)\lambda_1(t)}\dd s\\
&+\int_t^\infty e^{-r(s-t)+A_0(t,s)-A_1(t,s)\lambda_1(t)} \dd s.
}
\label{prop:singleOU}
\end{proposition}

To derive the optimal solution, we define a new equivalent probability measure $\Pbt$. The complete proof is given in Appendix \ref{ap:singleOU}.

\begin{proposition}
Suppose that $\lambda_1(t)$ follows the CIR process \eqref{eq:singleCIR}. Then, we have
\eqstar{
\Eb_t\left[e^{-\int_t^s\lambda_1(u)du}\right]=e^{A_0(t,s)-A_1(t,s)\lambda_1(t)},
}
where 
\eqstar{
A_1(t,s)& = \frac{2(e^{\eta(s-t) }-1)}{(b_1 +\eta)(e^{\eta(s-t)}-1)+2\eta},\qquad \eta=\sqrt{b_1 ^2 +2\sigma_1^2},\nonumber\\
A_0(t,s)& = -\int_t^s a_1(u)A_1(u,s)\dd u.
}
Suppose $\theta_1(t,\lambda_1):=\theta_1\sqrt{\lambda_1(t)}.$ Then, the solution to the single-population optimisation problem \eqref{eq:problem} is given as
\eqstar{
\frac{\beta^*(t)}{Y(t)}= \frac{1}{G(t,\lambda_1)}, && \frac{\alpha_S^*(t)}{Y(t)}= \frac{\theta_S}{\sigma_S},&& \frac{\alpha_L^*(t)}{Y(t)} = \frac{\theta_1\sqrt{\lambda_1(t)}}{\sigma_L(t)}+\frac{\sigma_1\sqrt{\lambda_1(t)}}{\sigma_L(t)}\frac{G_{\lambda_1}(t,\lambda_1)}{G(t,\lambda_1)},
}
where
\eqstar{
\sigma_L(t) = &-A_1(t,t+T)\sigma_1\sqrt{\lambda_1(t)},\\
G(t,\lambda_1) = &\ \phi\lambda_1(t)\int_t^\infty e^{-(b_1+r)(s-t)+A_0(t,s)-A_1(t,s)\lambda_1(t)-\sigma_1^2\int_t^s A_1(u,s)\dd u}\dd s\\
&+\phi\int_t^\infty \int_t^s a_1(u)e^{-b_1(s-u)-\sigma_1^2\int_u^s A_1(v,s)\dd v}\dd u \ e^{-r(s-t)+A_0(t,s)-A_1(t,s)\lambda_1(t)}\dd s\\
&+\int_t^\infty e^{-r(s-t)+A_0(t,s)-A_1(t,s)\lambda_1(t)} \dd s.
}
\label{prop:singleCIR}
\end{proposition}
As the proof of above result is similar to the proof of Proposition \ref{prop:singleOU}, we omit it here to avoid repetition.

\vspace{0.5pc}

We learn from Proposition \ref{prop:singleOU} and \ref{prop:singleCIR} that the optimal portfolio weight on stock under OU and CIR processes are the same and keep constant over time. However, it is difficult to see from the solution how the longevity bond investment changes over time, because $\frac{\alpha^*_L(t)}{Y(t)}$ depends on $G(t,\lambda_1)$ and $G_{\lambda_1}(t,\lambda_1)$. Later in Section \ref{sec:numerical}, we perform numerical simulations to investigate the optimal investment strategy and benefit withdrawal rate dynamically.

\subsection{Sub-population model} 
Based on different mortality indices and maturity times, there may be different longevity bonds in the market. It may be interesting to study our problem in a market setting with multiple longevity bonds. However, we only consider one longevity bond in our framework as our main focus is on longevity risk hedging. In this section, we study the case where the longevity bond's reference population is different from the pension members' population. More specifically, we set $n=2$ and denote the index population of the longevity bond by Population 1, and use Population 2 to denote the pension members' population. In practice, the reference population of longevity bond tends to be large and could be much larger than pension schemes' reference population. Therefore, we assume that the pension members are a sub-population of the index population. Under this sub-population assumption (see, \cite{wong2014time}), vectors $\mathscr{A}(t,\lambda)$ and $\Sigma^\prime(t,\lambda)$ in \eqref{eq:lambdas} take the following form:
\eqlnostar{eq:double}{
&\mathscr{A}(t,\lambda)=\left[\begin{array}{ccc}
    a_1(t)+b_1\lambda_1(t)\\
    a_2(t)+b_{21}\lambda_1(t)+b_{22}\lambda_2(t)
    \end{array} \right],& &\Sigma^\prime (t,\lambda)= \left[\begin{array}{ccccc}
    \sigma_1(t,\lambda_1)  &  0 &0  \\
    \sigma_{21}(t,\lambda_1)  & \sigma_{22}(t,\lambda_2) & 0 
    \end{array} \right],&
}
where $b_1$, $b_{21}$ and $b_{22}$ are constant numbers.
\eqstar{
a_i(t) = b_i\left(\nu_i+\frac{1}{\Del_i}\left( 1+\frac{1}{b_i \Del_i}\right)e^{\frac{t-m_i}{\Del_i}}\right), \quad \text{for}\  i = 1,2,
}
where $\nu_i$ are $\Del_i$ constants. The constants $m_1$ and $m_2$ are modal values of life expectancy of Population 1 and Population 2, respectively.

Under OU and CIR processes, the continuous functions $\sigma_1(t,\lambda_1)$, $\sigma_{21}(t,\lambda_1)$ and $ \sigma_{22}(t,\lambda_2)$ in vector $\Sigma^\prime(t,\lambda)$ are given as:
\begin{itemize}
\item{ OU process: }
\eqlnostar{eq:doubleOU}{
\Sigma^\prime (t,\lambda)&= \left[\begin{array}{ccc}
    \sigma_1  &  0 &0  \\
    \sigma_{21}  & \sigma_{22} & 0 
    \end{array} \right].
}
\item{ CIR process:}
\eqlnostar{eq:doubleCIR}{
\Sigma^\prime (t,\lambda)&= \left[\begin{array}{ccc}
    \sigma_1\sqrt{\lambda_1(t)}  &  0 &0  \\
    \sigma_{21}\sqrt{\lambda_1(t)}  & \sigma_{22}\sqrt{\lambda_2(t)} & 0 
    \end{array} \right].
}
\end{itemize}
In this sub-population model, there are two state variables. Moreover, the state variable $\lambda_2(t)$ correlates with the state variable $\lambda_1(t)$. This increases the difficulty of solving the optimisation problem. In this case, an analytical solution may not always be available. In the following, we provide Proposition \ref{prop:doubleOU} and \ref{prop:doubleCIR} for the solutions to the sub-population optimisation problem under OU and CIR settings, respectively.

\begin{proposition}
Suppose Population 2 is a sub-population of Population 1, and the force of mortalities follow OU processes: the dynamics of $\lambda_1(t)$ and $\lambda_2(t)$ are described as
\eqstar{
\dd \lambda_1(t) =& \Big(a_1(t)+b_1\lambda_1(t)\Big)\dd t+\sigma_1\dd W_1(t),\\
\dd \lambda_2(t) =& \Big(a_2(t)+b_{21}\lambda_1(t)+b_{22}\lambda_2(t)\Big)\dd t+\sigma_{21}\dd W_1(t)+\sigma_{22}\dd W_2(t).
}
Then, we have
\eqstar{
&\Eb_t\left[e^{-\int_t^s\lambda(u)du}\right]=e^{A_0(t,s)-A_1(t,s)\lambda_1(t)},\nonumber\\
&\Eb_t\left[e^{-\int_t^s\lambda_2(u)du}\right]=e^{C_0(t,s)-C_1(t,s)\lambda_1(t)-C_2(t,s)\lambda_2(t)},
}
where functions $A_0(t,s)$ and $A_1(t,s)$ are the same as in Proposition \ref{prop:singleOU}. The functions $C_0(t,s)$, $C_1(t,s)$ and $C_2(t,s)$ are
\eqstar{
C_2(t,s) =& \frac{1}{b_{22}}\left(1-e^{-b_{22}(s-t)}\right),\nonumber\\
C_1(t,s) =& \frac{b_{21}}{b_1(b_1-b_{22})}\left(1-e^{-b_1(s-t)}\right)+\frac{b_{21}}{b_{22}(b_1-b_{22})} ,\nonumber\\
C_0(t,s) =& -\int_t^s\Big[a_1(u)C_1(u,s)+a_2(u)C_2(u,s)-\frac{1}{2}\sigma_1^2 C_1^2(u,s) \\
&-\frac{1}{2}(\sigma_{21}^2+\sigma_{22})^2 C_2^2(u,s)-\sigma_1\sigma_{21}C_1(u,s)C_2(u,s) \Big]\dd u.}
Suppose $\theta_1(t,\lambda_1):=\theta_1.$ Then, the solution to the sub-population optimisation problem is given as
\eqstar{
\frac{\beta^*(t)}{Y(t)}= &\frac{1}{G(t,\lambda_1,\lambda_2)},\nonumber\\
\frac{\alpha_S^*(t)}{Y(t)}= &\frac{\theta_S}{\sigma_S} ,\nonumber\\
\frac{\alpha_L^*(t)}{Y(t)} = &\frac{\theta_1}{\sigma_L(t)}+\frac{\sigma_1}{\sigma_L(t)}\frac{G_{\lambda_1}(t,\lambda_1,\lambda_2)}{G(t,\lambda_1,\lambda_2)} +\frac{\sigma_{21}}{\sigma_L(t)}\frac{G_{\lambda_1}(t,\lambda_1,\lambda_2)}{G(t,\lambda_1,\lambda_2)},
}
where
\eqstar{
\sigma_L(t) = &-A_1(t,t+T)\sigma_1,\\
G(t,\lambda_1,\lambda_2) =& \ \phi \frac{b_{21}}{b_1-b_{22}}\lambda_1(t) \int_t^\infty \left(e^{-b_1(s-t)}-e^{-b_{22}(s-t)}\right)f(t,s)\dd s\\
&+\phi\lambda_2(t)\int_t^\infty e^{-b_{22}(s-t)}f(t,s)\dd s-\phi \int_t^\infty \Gamma_1(t,s)f(t,s) \dd s\\
& +\phi\frac{b_{21}}{b_1-b_{22}}\int_t^\infty \Gamma_2(t,s)f(t,s)\dd s+\int_t^\infty f(t,s) \dd s,\\
f(t,s) = & e^{-r(s-t)+C_0(t,s)-C_1(t,s)\lambda_1(t)-C_2(t,s)\lambda_2(t)},\\
\Gamma_1(t,s) = & \int_t^s e^{-b_{22}(s-u)}\Big[\frac{b_{21}}{b_1-b_{22}}a_1(u)-a_2(u)+\sigma_1\sigma_{21}C_1(u,s)\\
&+\left(\sigma_{21}^2+\sigma_{22}^2-\frac{b_{21}\sigma_1\sigma_{21}}{b_1-b_{22}}\right)C_2(u,s)\Big]\dd u,\\
\Gamma_2(t,s) = & \int_t^s e^{-b_1(s-u)}\left[a_1(u)-\sigma_1^2 C_1(u,s)-\sigma_1\sigma_{21} C_2(u,s)\right]\dd u.
}
\label{prop:doubleOU}
\end{proposition}
\begin{proof}
The proof is given in Appendix \ref{ap:doubleOU}
\end{proof}

\begin{proposition}
Suppose Population 2 is a sub-population of Population 1, and the force of mortalities follow CIR processes: the dynamics of $\lambda_1(t)$ and $\lambda_2(t)$ are described as
\eqstar{
\dd \lambda_1(t) =& \Big(a_1(t)+b_1\lambda_1(t)\Big)\dd t+\sigma_1\sqrt{\lambda_1(t)}\dd W_1(t),\\
\dd \lambda_2(t) =& \Big(a_2(t)+b_{21}\lambda_1(t)+b_{22}\lambda_2(t)\Big)\dd t+\sigma_{21}\sqrt{\lambda_1(t)}\dd W_1(t)+\sigma_{22}\sqrt{\lambda_2(t)}\dd W_2(t).
}
Then, we have
\eqstar{
&\Eb_t\left[e^{-\int_t^s\lambda(u)du}\right]=e^{A_0(t,s)-A_1(t,s)\lambda_1(t)},\nonumber\\
&\Eb_t\left[e^{-\int_t^s\lambda_2(u)du}\right]=e^{C_0(t,s)-C_1(t,s)\lambda_1(t)-C_2(t,s)\lambda_2(t)},
}
where functions $A_0(t,s)$ and $A_1(t,s)$ are the same as in Proposition \ref{prop:singleCIR}. The functions $C_0(t,s)$, $C_1(t,s)$ and $C_2(t,s)$ are solutions to the following ODE system
\eqlnostar{ODEs}{
0 & = -\frac{C_0}{\del t}+a_1 C_1+a_2 C_2,\nonumber\\
0 & =-\frac{\del C_1}{\del t}+b_1 C_1+b_{21} C_2+\frac{1}{2}\sigma_1^2 C_1^2+\frac{1}{2}\sigma_{21}^2 C_2^2+\sigma_1\sigma_{21} C_1\Act_2,\nonumber\\
0 & = -\frac{\del  C_2}{\del t}+b_{22} C_2+\frac{1}{2}\sigma_{22}^2 C_2^2-1,
}
with terminal conditions $C_0(s,s)=0$, $C_1(s,s)=0$ and $C_2(s,s)=0$.

Suppose $\theta_1(t,\lambda_1):=\theta_1\sqrt{\lambda_1(t)}.$ Then, the solution to the sub-population optimisation problem is given as
\eqstar{
\frac{\beta^*(t)}{Y(t)}= &\frac{1}{G(t,\lambda_1,\lambda_2)},\nonumber\\
\frac{\alpha_S^*(t)}{Y(t)}= &\frac{\theta_S}{\sigma_S} ,\nonumber\\
\frac{\alpha_L^*(t)}{Y(t)} = &\frac{\theta_1\sqrt{\lambda_1(t)}}{\sigma_L(t)}+\frac{\sigma_1\sqrt{\lambda_1(t)}}{\sigma_L(t)}\frac{G_{\lambda_1}(t,\lambda_1,\lambda_2)}{G(t,\lambda_1,\lambda_2)} +\frac{\sigma_{21}\sqrt{\lambda_1(t)}}{\sigma_L(t)}\frac{G_{\lambda_1}(t,\lambda_1,\lambda_2)}{G(t,\lambda_1,\lambda_2)},
}
where
\eqstar{
\sigma_L(t) = &-A_1(t,t+T)\sigma_1\sqrt{\lambda_1(t)}.
}

$G(t,\lambda_1,\lambda_2)$ satisfies
\eqstar{
G(t,\lambda_1,\lambda_2) =&\int_t^\infty e^{-r(s-t)+C_0(t,s)-C_1(t,s)\lambda_1(t)-C_2(t,s)\lambda_2(t)}\dd s +\phi\int_t^\infty e^{-r(s-t)}\Ebt_t\left[\lambda_2(s)\right]\dd s.
}
$\Ebt[\cdot]$ denotes the expectation under an equivalent measure $\Pbt$ which is defined as
\eqstar{
\frac{\dd \Pbt}{\dd \Pb}=\Zt(s)=\exp\left(-\int _0^s \tilde{\theta}^\prime(u,s)\dd W(u)-\frac{1}{2}\int_0^s \left\|\tilde{\theta}(u,s)\right\|^2\dd u\right),
}
where
\eqstar{
\tilde{\theta}(u,s) = &\left[\begin{array}{ccc} 
\sigma_1\sqrt{\lambda_1(u)} C_1(u,s)+\sigma_{21}\sqrt{\lambda_1(u)} C_2(u,s) \\
 \sigma_{22} \sqrt{\lambda_2(u)}C_2(u,s) \\
  0
\end{array}\right].
}
\label{prop:doubleCIR}
\end{proposition}
\begin{proof}
The Picard-Lindelöf theorem ensures the existence of unique solutions to ODEs in \eqref{ODEs}. The rest of the proof is similar to the proof of  Proposition \ref{prop:doubleOU} and is omitted here.
\end{proof}

From Proposition \ref{prop:doubleOU} and \ref{prop:doubleCIR}, we see that the optimal portfolio weight on stock stays the same as in the single-population case. The OU case has an analytical solution; however, the CIR case does not. In Section \ref{sec:numerical}, we conduct numerical studies to observe the hedging effect of the longevity bond in the sub-population OU process case.

\section{Numerical applications}
\label{sec:numerical}
This section provides numerical simulations in the single- and sub-population cases using the results from Proposition \ref{prop:singleOU} and \ref{prop:doubleOU}. A numerical study involving CIR models is not presented here as the results are similar to the OU case. We observe the dynamics of the survival probability and look into the impact of the mortality behaviour on the optimal strategy. We investigate the hedging performance of the longevity bond in the pension scheme's risk management and provide a sensitivity analysis on the market price of longevity risk. We are also interested in the effect of the risk-sharing rule between the members and the manager.

Table \ref{table:base} shows the values of parameters in our numerical examples. The time interval $\Del=1/10$ means that we observe the mortality rates 10 times a year. Most of the values of the mortality model parameters are as considered in other works (e.g. \citep{menoncin2017longevity} and \cite{milevsky2001optimal}). The values of other financial market parameters are meant to be representative.
\begin{table}[htbp]
\begin{center}
\small
\caption{Values of parameters for optimization problem.}
{\begin{tabular}{llll}
\toprule
Population 1 & Population 2 & Market & Others\\
\toprule
$\nu_1=0.0009944$ & $\nu_2=0.0009944$ & $r = 0.04$  &  $T=35$ \\
$\Del_1 = 11.4000$  & $\Del_2= 12.9374$& 	$\theta_1=-0.0005$    &  $Y_0=100$\\
$m_1=86.4515$ & $m_2=89.18$ & $\theta_S=0.05$  &  $\Del=1/10$\\
$b_1=0.5610$ & $b_{21}=0.0028$& $\sigma_S=0.15$   &$\phi=0.8$\\
$\sigma_1=0.0035$ & $b_{22}=0.6500$ & $T_L = 20$ &  \\
 & $\sigma_{21}=0.0040$ & &  \\
  & $\sigma_{22}=0.0050$ & &  \\
\hline
\end{tabular}}
\label{table:base}
\end{center}
\end{table}
%\begin{table}[htbp]
%\centering
%\small
%\caption{Values of parameters for optimization problem.}
%{\begin{tabular}{llll}
%\toprule
%Population 1 & Population 2 & Market & Others\\
%\toprule
%$\nu_1=0.0009944$ & $\nu_2=0.0009944$ & $r = 0.04$  &  $T=35$ \\
%$\Del_1 = 11.4000$  & $\Del_2= 12.9374$& 	$\theta_1=-0.0005$    &  $Y_0=100$\\
%$m_1=86.4515$ & $m_2=89.18$ & $\theta_S=0.05$  &  $\Del=1/10$\\
%$b_1=0.5610$ & $b_{21}=0.0028$& $\sigma_S=0.15$   &$\phi=0.8$\\
%$\sigma_1=0.0035$ & $b_{22}=0.6500$ & $T_L = 20$ &  \\
% & $\sigma_{21}=0.0040$ & &  \\
%  & $\sigma_{22}=0.0050$ & &  \\
%\hline
%\end{tabular}}
%\label{table:base}
%\end{table}

\subsection{Single-population case}
\label{subsec:numerical1}
In the single-population case, we assume that the scheme members happen to be the reference population of a longevity bond. The manager invests in this longevity bond to hedge the scheme's longevity exposure. We assume that the scheme members have similar mortality behaviour, and their retirement age is 65. The pension scheme's management starts from the retirement time and ends until the last member passes away. According to the Gompertz-Makeham law of mortality, less than $5\%$ of the population could survive until 100 years old given that they are alive at 65 years old. Thus, we conduct numerical simulations with a 35-year time horizon (i.e. $T=35$). 

Figure \ref{fig:sub} shows three simulation paths to illustrate the dynamic mortality behaviour of the members. We see that the survival probability decreases with time. In the bottom-right plot, we observe that for all three paths, less than $5\%$ of the members survive until 100 years old indicating that it is reasonable to set $T=35$. Suppose the Gompertz-Makeham mortality law describes the average trend of the members' mortality. We learn from the plots in Figure \ref{fig:sub} that the survival probability of the simulated path 1 is always higher than average. On the contrary, the survival probability of path 2 is lower than expected. Path 3 does not show any particular trend. Figure \ref{fig:sub} shows the cumulative distribution function $f(\cdot)$ of the stochastic death time $\tau$. We see from the plot that $f(\tau)$ peaks at $\tau=86.5$ approximately. This is consistent with our model settings and choice of parameter values: $m_1=86.4515$ in our model is the modal value of life span of the members. Generally, our observations imply that:
\begin{itemize}
\item path 1 shows a prominent longevity trend;
\item path 2 on average has a lower survival rate;
\item path 3 does not show any particular trend;
\item For all three paths, most of the members' deaths happen before the age of 86.5.
\end{itemize}

\begin{figure}[htbp]
\centering
    \includegraphics[scale=0.55]{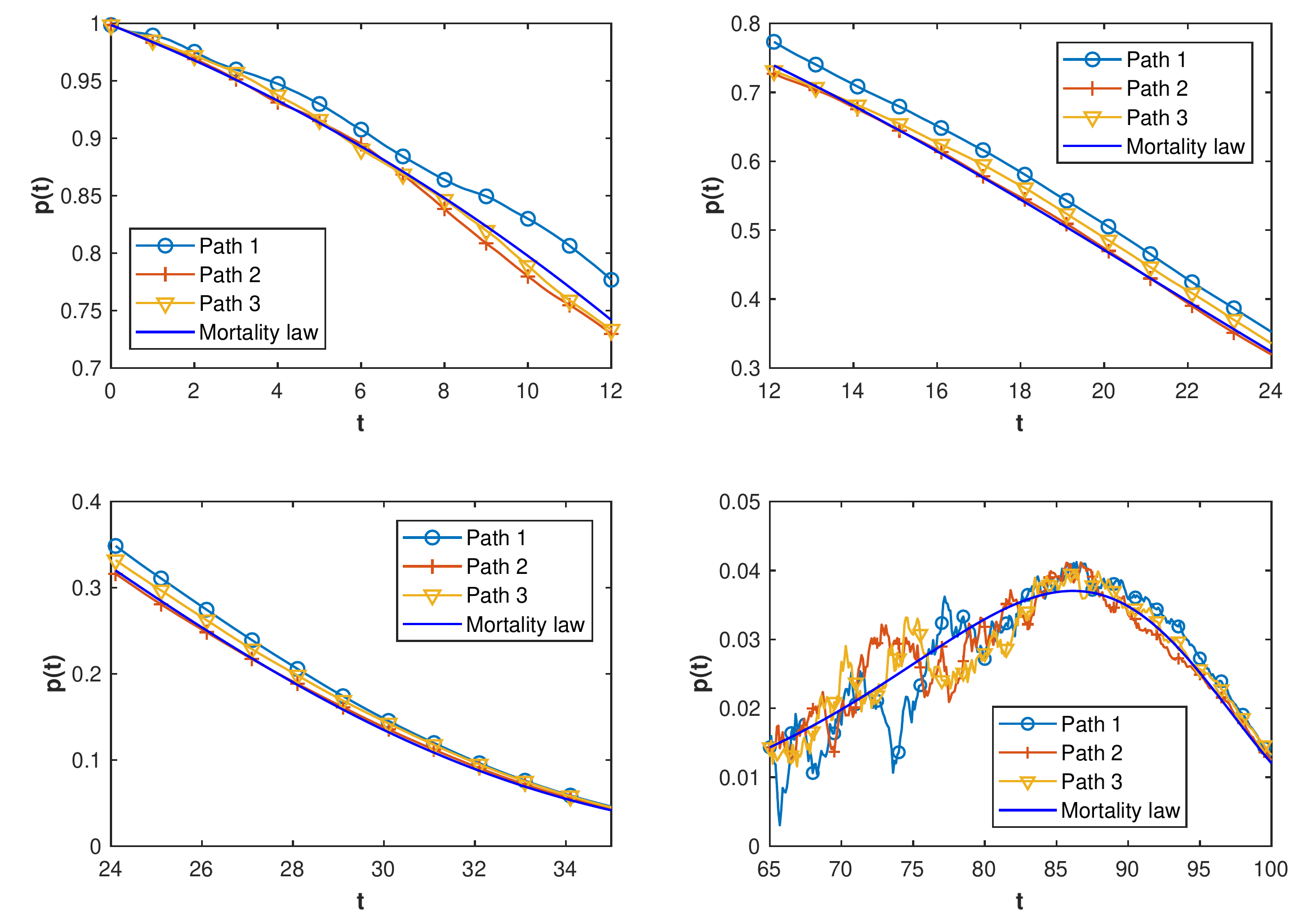}
    \caption{\textit{3 simulation paths for the survival probability and cumulative distribution function of $\tau$}}
    \label{fig:sub}
\end{figure}

\subsubsection{The base scenario}
In the base scenario, we investigate the optimal benefit withdrawal rate and investment strategy. Besides, we observe the scheme's wealth level and manager's compensation dynamically. Figure \ref{fig:strategy} shows the average investment strategy over 100 simulation paths. We observe that the portfolio weight in the longevity bond drops over time. As members get older, the scheme's exposure to the longevity risk reduces, and the need for longevity protection decreases. Accordingly, the manager reduces the portfolio weight in the longevity bond. The flat line which shows the investment proportion in stock is a direct result from the optimal solutions in Proposition \ref{prop:singleOU}: the portfolio weight in stock keeps constant and equals to $\frac{\theta_S}{\sigma_S}$. This constant investment strategy coincides with the result in the classical Merton's portfolio problem. The interpretation is that the constant market price of risk causes no change in the manager's investment behaviour. The proportion of the portfolio in the money market is $\frac{\alpha_0(t)}{Y(t)}=1-\frac{\alpha_L(t)}{Y(t)}-\frac{\alpha_S(t)}{Y(t)}$. We see that the portfolio weight in the money market is negative at first and then increases over time. The negative position in the initial years indicates that the manager borrows money from the money market, and invests in the risky assets to gain risk premiums and increase the scheme's wealth level. As $\frac{\alpha_S(t)}{Y(t)}$ keeps constant, and $\frac{\alpha_L(t)}{Y(t)}$ decreases over time, the manager puts more weight in the money market. With the passage of time, the manager becomes more conservative to avoid unexpected losses. Overall, the longevity bond dominates the investment portfolio throughout the investment horizon. Even when members reach the age of 100, the manager puts around $50\%$ of the portfolio in the longevity bond, indicating that the longevity bond could provide not only longevity protection but also considerable risk premium. 

\begin{figure}[htbp]
\centering   \includegraphics[scale=0.5]{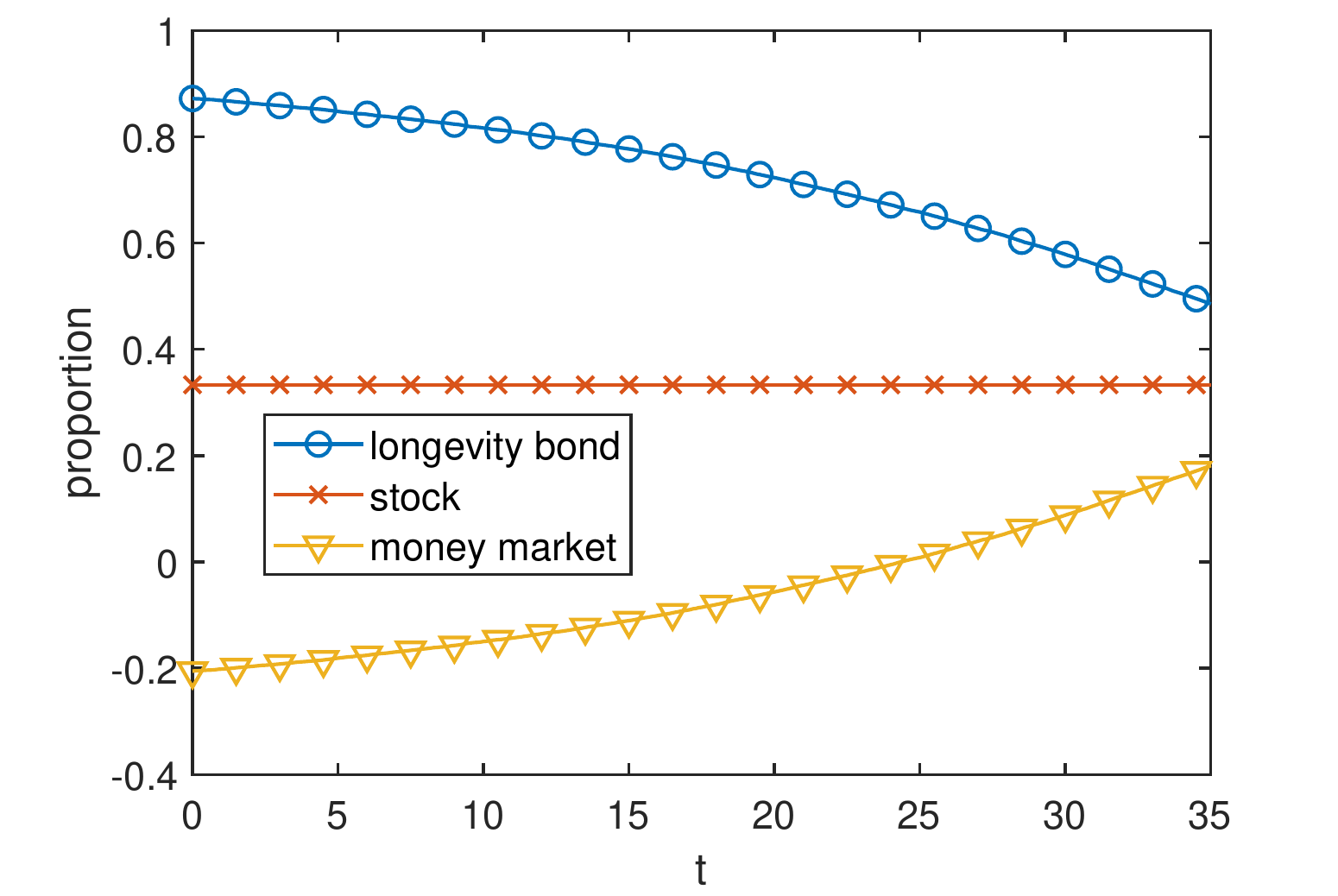}
   \caption{\textit{ average optimal investment strategies over 100 paths}}
    \label{fig:strategy}
    \end{figure}
%\begin{figure}[htbp]
%\centering
%\includegraphics[scale= 0.42]{benefit.pdf}
%\caption{\textit{simulation paths for the optimal benefit withdrawals}}
%\label{fig:benefit}
%\end{figure}
%\begin{figure}[htbp]
%\begin{center}
%    \includegraphics[scale=0.4]{wealth_compensation.pdf}
%    \caption{\textit{simulation paths for the wealth process and compensation}}
%    \label{fig:wealth}
%\end{center}
%\end{figure}

To observe general results, we show the average of 100 simulation paths of the optimal benefit withdrawal proportion and rate, the scheme's wealth and the manager's compensation. From the top-left plot in Figure \ref{fig:100}, we see the optimal proportion of the scheme wealth $\frac{\beta^*}{Y(t)}$ withdrawn by members increases over time. Meanwhile, in the top-right plot, we observe that the optimal benefit withdrawal rate $\beta^*(t)$ reduces over time. This phenomenon is explained by the scheme's declining wealth process shown in the bottom-right plot. Although the manager invests in the financial market, the average scheme level is decreasing throughout the time horizon due to the continuous benefit withdrawals and compensation payments. The decreasing number of survival members, as well as the declining scheme wealth, result in the drop in benefit withdrawal rate. The average compensation received by the manager shows an interesting trend - it increases at first, peaks at around the 19th year and then drops rapidly. The reason is that most members pass away before the age of 84. The compensation increases before the mode of the members' life span even though the scheme's wealth level keeps decreasing. After reaching the modal value of life expectancy, as most of the members have already passed away, the manger receives less compensation.

%In Figure \ref{fig:100}, we show the average of 100 simulation paths of the scheme's wealth, the optimal benefit withdrawal rate and the manager's compensation to observe general results. Although the manager invests in the financial market, the average scheme level is decreasing throughout the time horizon. It is due to the continuous benefit withdrawals and the compensation payments. The average benefit withdrawal rate also shows a declining trend. It is because of the decreasing number of survival members as well as the declining scheme wealth. The average compensation received by the manager shows an interesting trend: it increases at first, peaks at around the 19th year and then drops rapidly. The reason is that most of the members pass away before the age of 84. The compensation increases before the mode of the members' life span, though the scheme's wealth level keeps decreasing. After reaching the modal value of life expectancy, as most of the members have already passed away, the manger receives less compensation.  
\begin{figure}[htbp]
\begin{center}
    \includegraphics[scale=0.43]{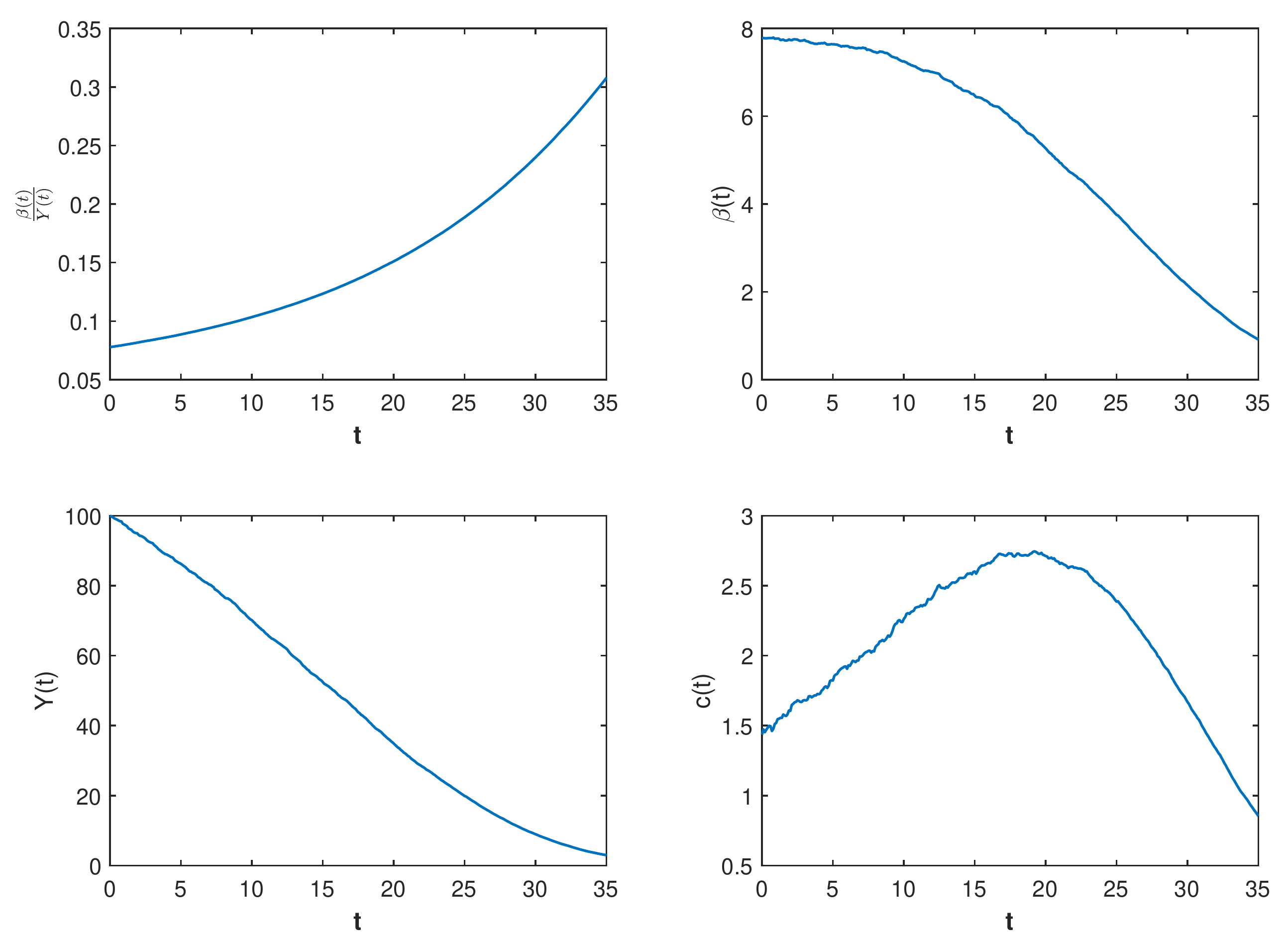}
    \caption{\textit{average benefit withdrawal proportion, withdrawal rate, wealth, and compensation over 100 paths}}
    \label{fig:100}
\end{center}
\end{figure}

\subsubsection{Comparison analysis}
\label{sec:comparison}
To test the hedging performance of longevity bond, we give the results in the case when the manager does not include the longevity bond in the investment portfolio. Without the investment in longevity bond, the optimal benefit withdrawal proportion $\frac{\beta^*(t)}{Y(t)}=\frac{1}{G(t,\lambda_1)},$ and the optimal portfolio weight on stock $\frac{\alpha^*_s(t)}{Y(t)}=\frac{\theta_S}{\sigma_S},$ are the same as given in Proposition \ref{prop:singleOU}. The portfolio weight in money market equals to $1-\frac{\alpha^*_S(t)}{Y(t)}$ and keeps constant over time. Let $\beta_1(t)$ ($c_1(t)$) and $\beta_2(t)$ ($c_1(t)$) denote the benefit withdrawal (compensation) without and with investment in longevity bond, respectively. Figure \ref{fig:benefit2} shows the benefit withdrawal and compensation improvement by investing in longevity bond. It shows that, for path 1, investing in longevity bond always results in higher benefit withdrawals and compensations. For path 3, investing in longevity bond in general increases both members' benefit withdrawals and the manager's compensation. Although, during some short period, longevity bond investment decreases the withdrawals and compensations. However, for path 2, investment in longevity bond seems to cut down both benefit withdrawals and compensations. As discussed earlier (see Figure \ref{fig:sub}), the survival probability on path 2 is overall lower than the Gompertz-Makeham survival probability. It implies that the pension members tend to live shorter than expected. Likewise, Figure \ref{fig:sub} suggests that for path 2, the random death age $\tau$ is more likely to be the younger ages compared to the other paths - path 1 or path 3. As a result, the pension scheme does not face the longevity risk and investing in the longevity bond actually loses money rather than making gains. As it is a global trend that people's average life expectancy is increasing, we argue that the pension schemes will benefit from longevity bond investment as mirrored by the situation in path 1 and 3. To support our claim, we show the average improvements of benefit withdrawal and compensation over 100 simulation paths in Figure \ref{fig:improvement}. As shown, there are low improvements in the first few years, but overall the improvements are significant over most of the investment period. It indicates that investing in longevity bond increases both members' benefit withdrawal rate and manager's compensation, and that it is beneficial to both members and manager. 
\begin{figure}[htbp]
\begin{center}
    \includegraphics[scale=0.47]{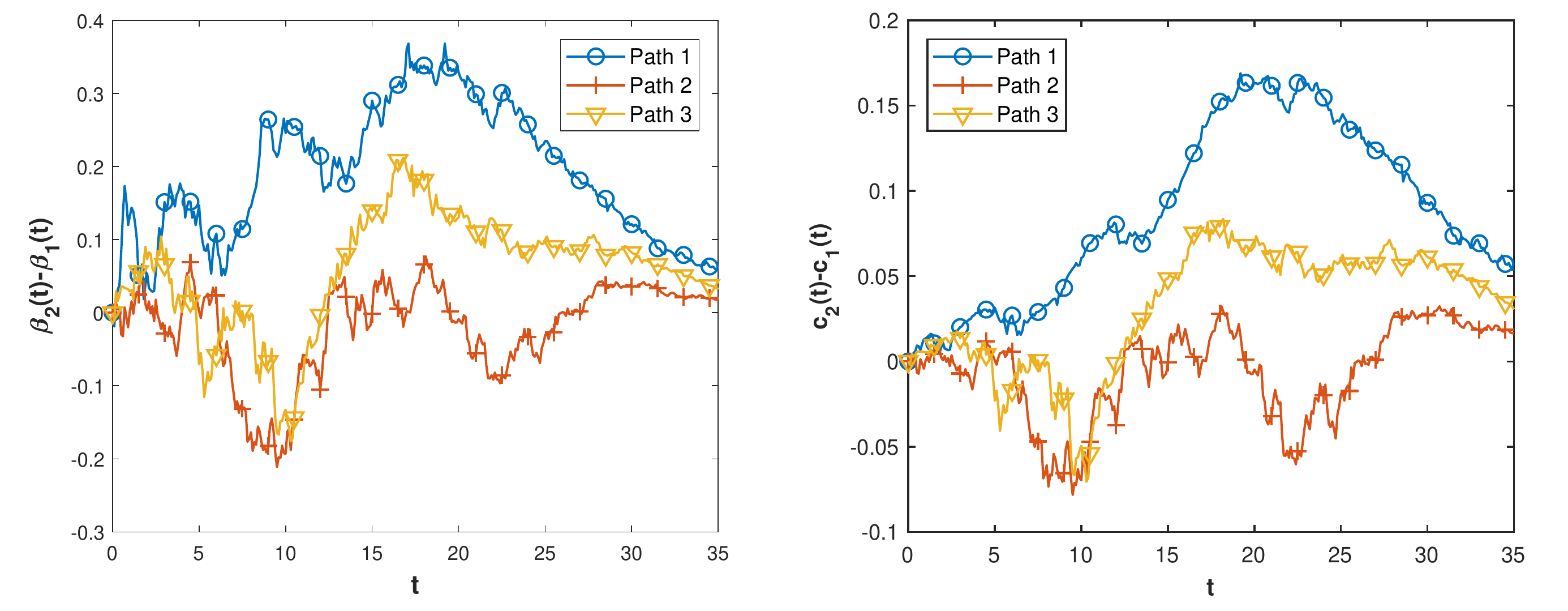}
    \caption{\textit{improvement of the benefit withdrawal and compensation}}
    \label{fig:benefit2}
\end{center}
\end{figure}
\begin{figure}[htbp]
\begin{center}
    \includegraphics[scale=0.5]{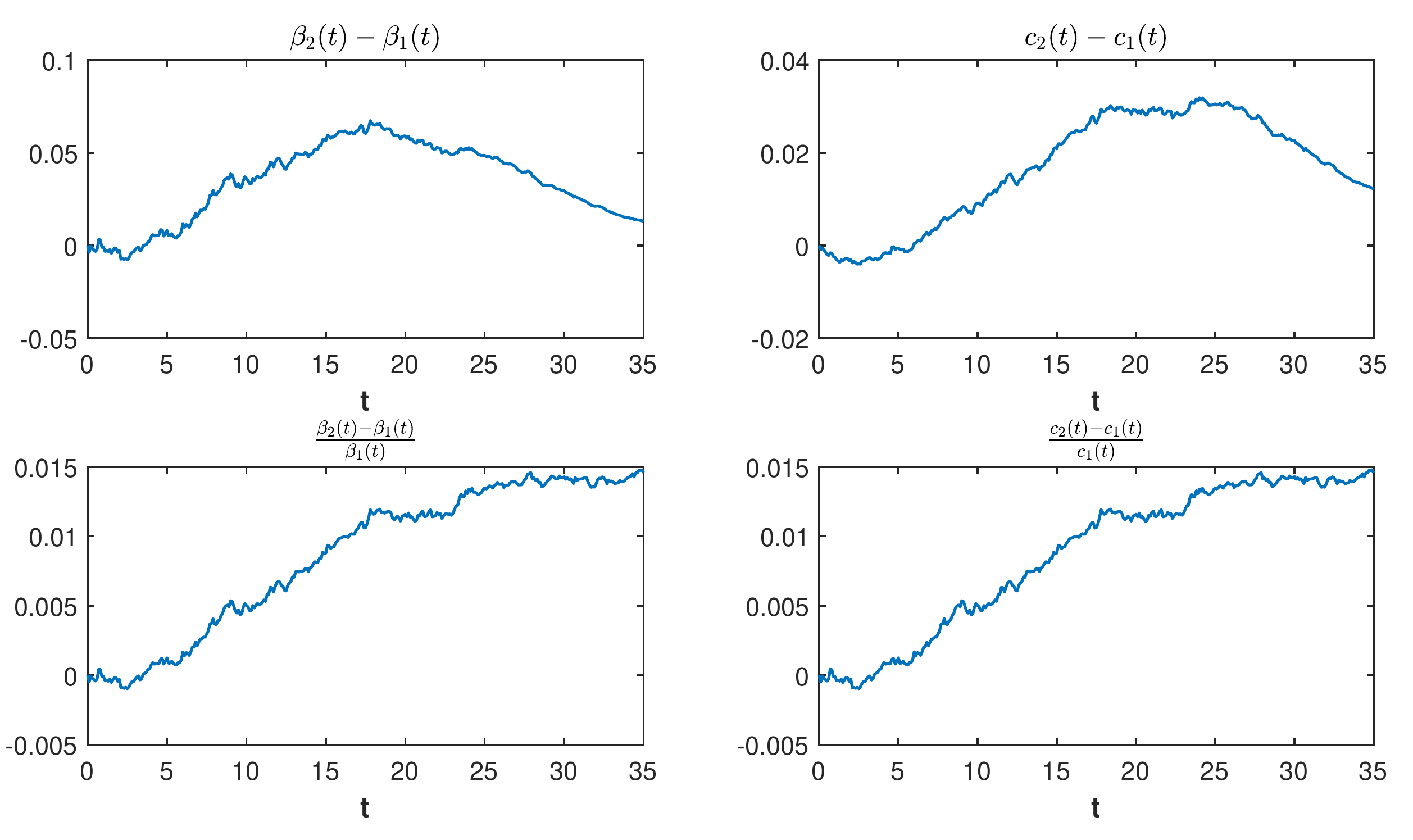}
    \caption{\textit{average improvements of benefit withdrawal and compensation over 100 paths}}
    \label{fig:improvement}
\end{center}
\end{figure}

\subsubsection{Sensitivity analysis}
We are interested in the impact of the market price of risk on pension scheme's risk management. It is difficult to decide the market price of longevity risk (in other words $\theta_1$) due to the absence of longevity bonds in the market. But we examine a few values of $\theta_1$ to give an illustration of the effects of the market price of longevity risk on the optimal strategy, the benefit withdrawal and the manager's compensation. In our base scenario, $\theta_1$ is set as $-5\times 10^{-4}$, and the longevity risk premium offered by the longevity bond is $4.4563\times 10^{-6}$. Compared with the stock's risk premium ($7.5\times 10^{-3}$), the longevity risk premium is minimal. We show the optimal investment strategies in the cases where $\theta_1=0$, $-1.5\times 10^{-3}$ and $-3\times 10^{-3}$. We know that a large absolute value of $\theta_1$ indicates high risk premium offered by the longevity bond. Besides, large premiums should lead to more investment in longevity bond. Figure \ref{fig:theta2} shows the investment strategies with different values of $\theta_1$. We can see that when the manager does not add longevity bond to his portfolio, the optimal investment proportions in stock and money market keep constant over time. The investment in the longevity bond does not affect the portfolio weight in stock while affects the investment in the money market. The top-right plot shows that even in the case when the longevity bond offers no risk premium (i.e. $\theta_1=0$), the proportion invested in the longevity bond is always higher than 40\%. It illustrates that the longevity bond provides a good way to hedge the scheme's longevity risk. As shown in Figure \ref{fig:theta2}, as the longevity risk premium decreases (i.e. lower $\theta_1$), more portfolio weight is put on the longevity bond. In the bottom-right plot, $\theta_1=-0.0030$ and the longevity risk premium equals to $2.6738\times 10^{-5}$ which is again far less than the stock's risk premium. It implies that the manager continuously borrows money from the money market to invest in risky assets throughout the whole time horizon. The intuition is that the longevity bond not only provides longevity risk hedge, but also provides an attractive risk premium. 
\begin{figure}[htbp]
\begin{center}
    \includegraphics[scale=0.53]{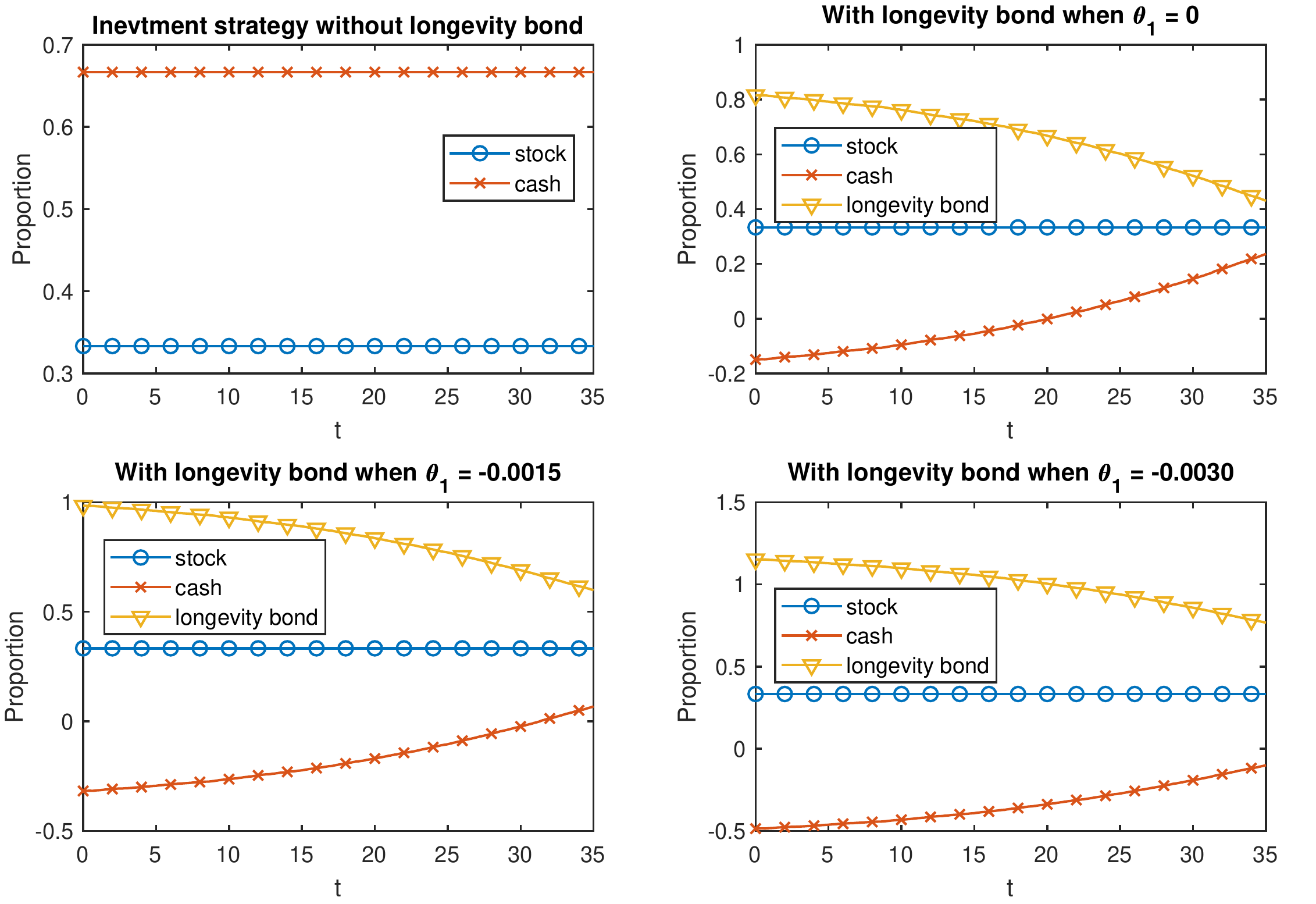}
    \caption{\textit{impact of $\theta_1$ on the optimal investment strategy}}
    \label{fig:theta2}
\end{center}
\end{figure}

Figure \ref{fig:theta} shows the improvements of benefit withdrawal and compensation when investing in longevity bond. As shown for path 1, high market prices of longevity risk lead to high improvements in both manager's compensation and members' benefit withdrawal. As discussed in Section \ref{sec:comparison}, investing in longevity bond sometimes decreases the benefit withdrawal and compensation and thus causes a `loss'. Here by `loss', we mean loss in the members benefit and manager compensation as the improvements by investing in longevity bond are negative. Nevertheless, we observe from the plots in Figure \ref{fig:theta} that a smaller $\theta_1$ reduces this loss.
\begin{figure}[htbp]
\begin{center}
    \includegraphics[scale=0.5]{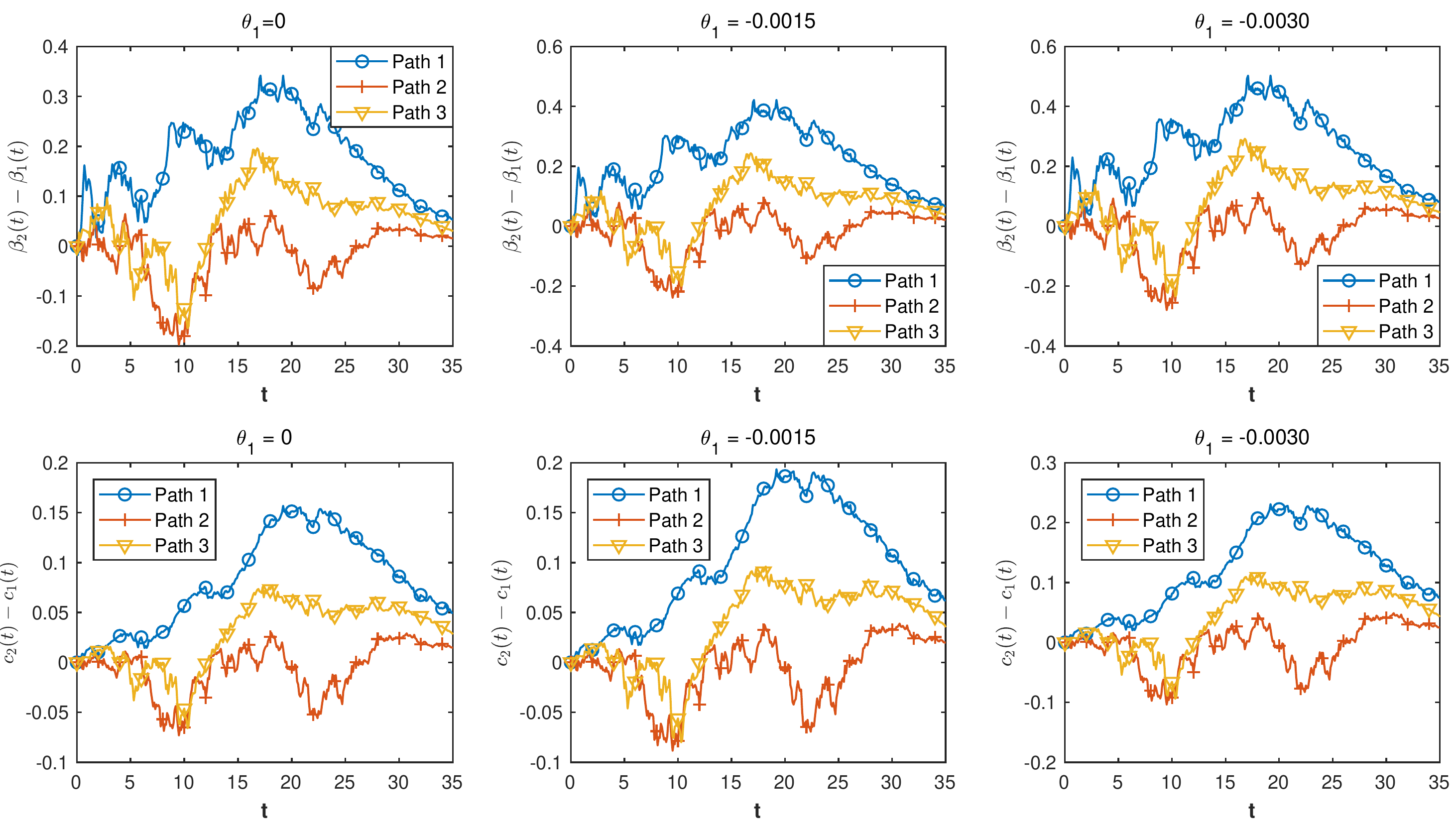}
    \caption{\textit{impact of $\theta_1$ on the benefit withdrawal and compensation improvements}}
    \label{fig:theta}
\end{center}
\end{figure}

We now investigate the impact of the risk-sharing parameter $\phi$ on the benefit withdrawal and manager compensation. As stated before, $\phi$ decides the agreement between the manager and members on how to share the risk. A high value of $\phi$ implies that more weight is put on the manager's utility. When $\phi=0$, the manager works on behalf of the members and only cares about the members' benefit. This case corresponds to no risk-sharing. When $0<\phi<1$, more attention is put on the members' benefit.  In the case where $\phi=1$, the manager treats his own profit and members' benefit equally which corresponds to the case of equal risk-sharing. We test the cases when $\phi$ takes values 0, 0.5 and 1. The case with $\phi=0$ is chosen as the reference case. Figure \ref{fig:phi} shows the improvement rates on benefit withdrawal and compensation (i.e. $\frac{\beta_\phi(t)-\beta_{\phi=0}(t)}{\beta_{\phi=0}(t)}$ and $\frac{c_\phi(t)-c_{\phi=0}(t)}{c_{\phi=0}(t)}$). As shown in the right plot, higher the $\phi,$ higher the improvement in compensation. Compared to the case $\phi=0$, the equal risk-sharing rule agreement improves the manager's compensation by more than $20\%$ at the end of the time horizon. For the benefit withdrawal, a higher $\phi$ leads to higher withdrawals in the last 10 years, but reduces the withdrawals in the early periods. To illustrate the impact of $\phi$, we calculate the average discounted values of benefit withdrawals and compensations over 100 simulation paths. We find that compared to the case $\phi=0$, $\phi=1$ increases the discounted benefit withdrawal by $4.71\%$ and the discounted compensation by $12.82\%$. Thus, both manager and members benefit from sharing the risk equally.
\begin{figure}[htbp]
\begin{center}
    \includegraphics[scale=0.47]{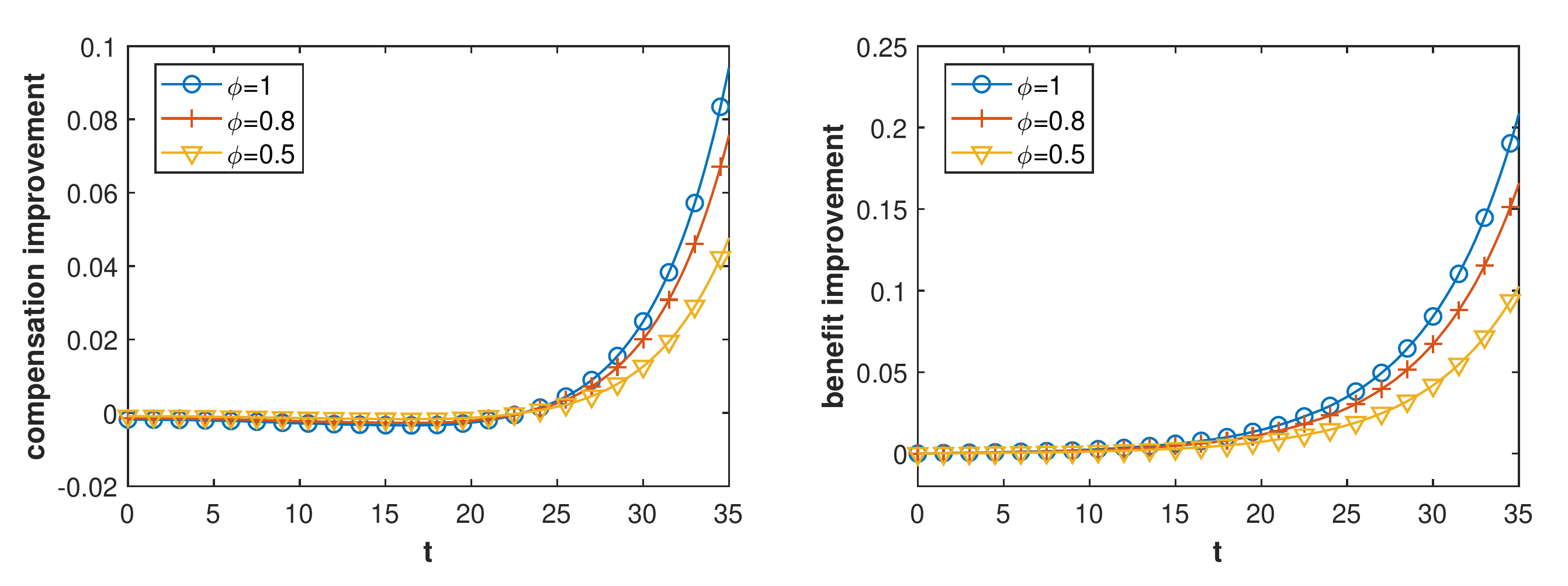}
    \caption{\textit{impact of $\phi$ on the benefit withdrawal and compensation}}
    \label{fig:phi}
\end{center}
\end{figure}

The previous discussions can be summarised as follows. The longevity bond offers an efficient way to hedge the pension scheme's longevity risk as it always dominates the investment portfolio. Moreover, both members and manager benefit from investment in the longevity bond. Higher the longevity risk premium, more portfolio weight is put on the longevity bond, and more the manager and members benefit from investing in it. Finally, an equal risk-sharing rule is the most beneficial to both members and scheme's manager.

\subsection{Sub-population case}
On the one hand, a pension scheme faces the longevity risk caused by its members' more extended life span.  On the other hand, a pension scheme faces the longevity basis risk if the mortality behaviour of the scheme members is imperfectly correlated with the longevity bond mortality index. In practice, pension schemes face longevity basis risk as it is difficult to find a longevity bond in the market that is based exactly on the scheme members. Therefore, the sub-population model may be more practical compared to the single-population model. In this section, we assume that the longevity bond is based on a large population and the scheme members are a sub-population of this large population. Furthermore, we assume that the life expectancy of the scheme members is higher than the longevity bond reference population. If it is lower, the pension scheme may not face the longevity risk and there is no need to invest in the longevity bond. We are interested in the optimal strategy and the hedging performance of longevity bond in the presence of longevity basis risk.

\subsubsection{The base scenario}
The values of parameters used in this section are as given in Table \ref{table:base}. For $i=1,2$, $m_i$ in \eqref{eq:double} is the mode of life expectancy of the $i$th population. $m_2$ is set to be greater than $m_1$ implying that the possible death age of the pension members of Population 2 is higher than the longevity bond reference population  in Population 1. Again, we present three simulation paths in this section to observe the two populations' mortality behaviour. Figure \ref{fig:survival2} shows the survival probabilities for Population 1 and Population 2. As expected, the average survival probability of Population 1 is lower than Population 2. The plots below illustrate the following:
\begin{itemize}
\item Population 1
    \begin{itemize}
    \item path 1 has no particular trend;
    \item path 2 has a higher survival probability than expected;
    \item path 3 shows shorter life expectancy.
    \end{itemize}
\item Population 2: 
    \begin{itemize}
    \item path 1 displays a lower survival probability in the first half of the time horizon;
    \item path 2 $\&$ 3 on average show longevity trend.
    \end{itemize}
\end{itemize}

\begin{figure}[htbp]
\begin{center}
    \includegraphics[scale=0.43]{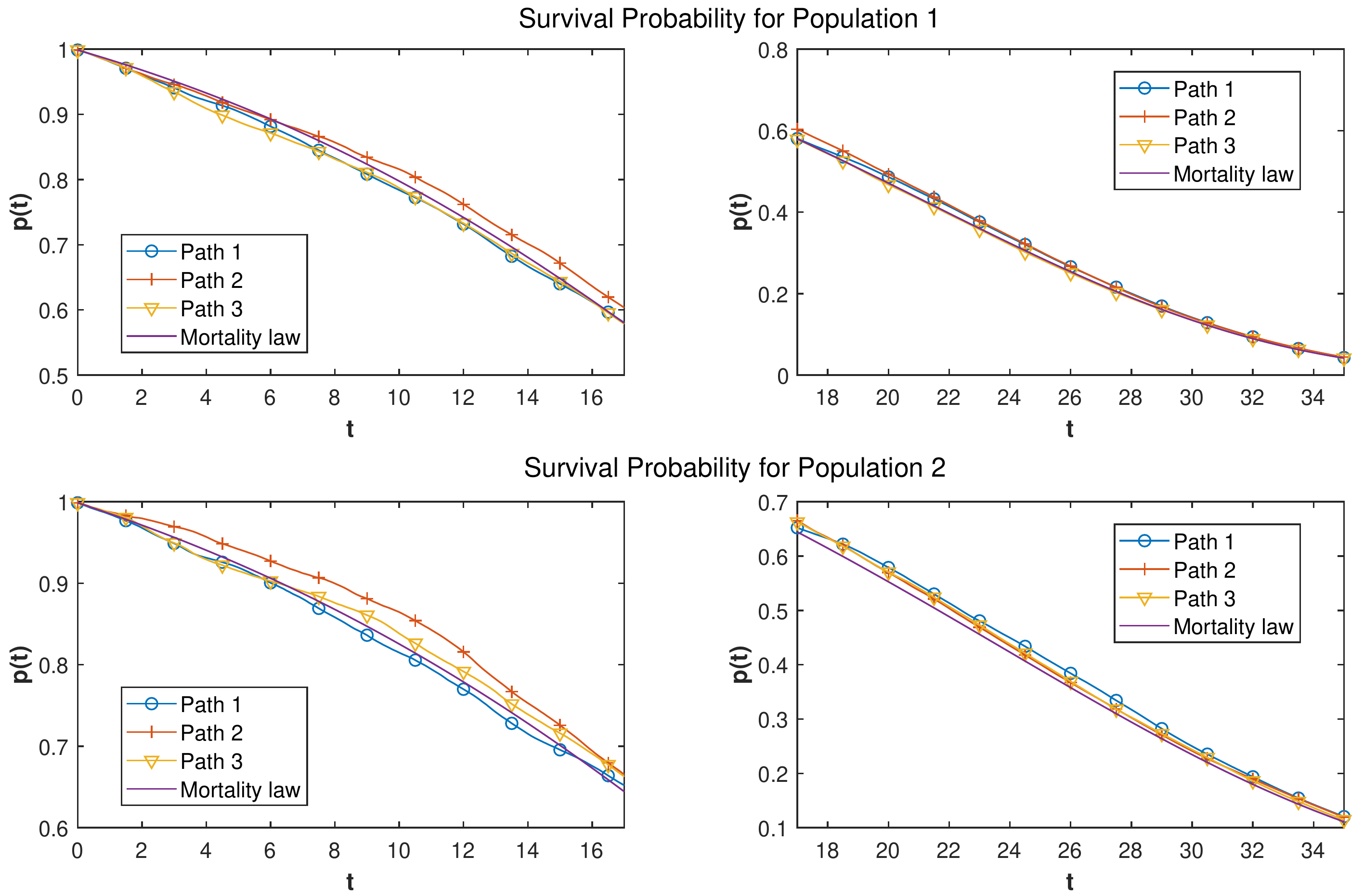}
    \caption{\textit{survival probabilities for Population 1 and 2}}
    \label{fig:survival2}
\end{center}
\end{figure}

Figure \ref{fig:doublestrategy} shows the optimal investment strategy and benefit withdrawal in the sub-population case. The optimal strategy is similar to the single-population case:
\begin{itemize}
\item the longevity bond dominates the portfolio;
\item the portfolio weight on the longevity bond decreases over time while the holding in the money market increases;
\item the optimal investment proportion in the stock keeps constant;
\item the percentage of the funds withdrawn by the members increases with time.
\end{itemize} 
We learn from it that the longevity bond plays an essential role in the pension scheme's risk management. 
\begin{figure}[htbp]
\begin{center}
    \includegraphics[scale=0.42]{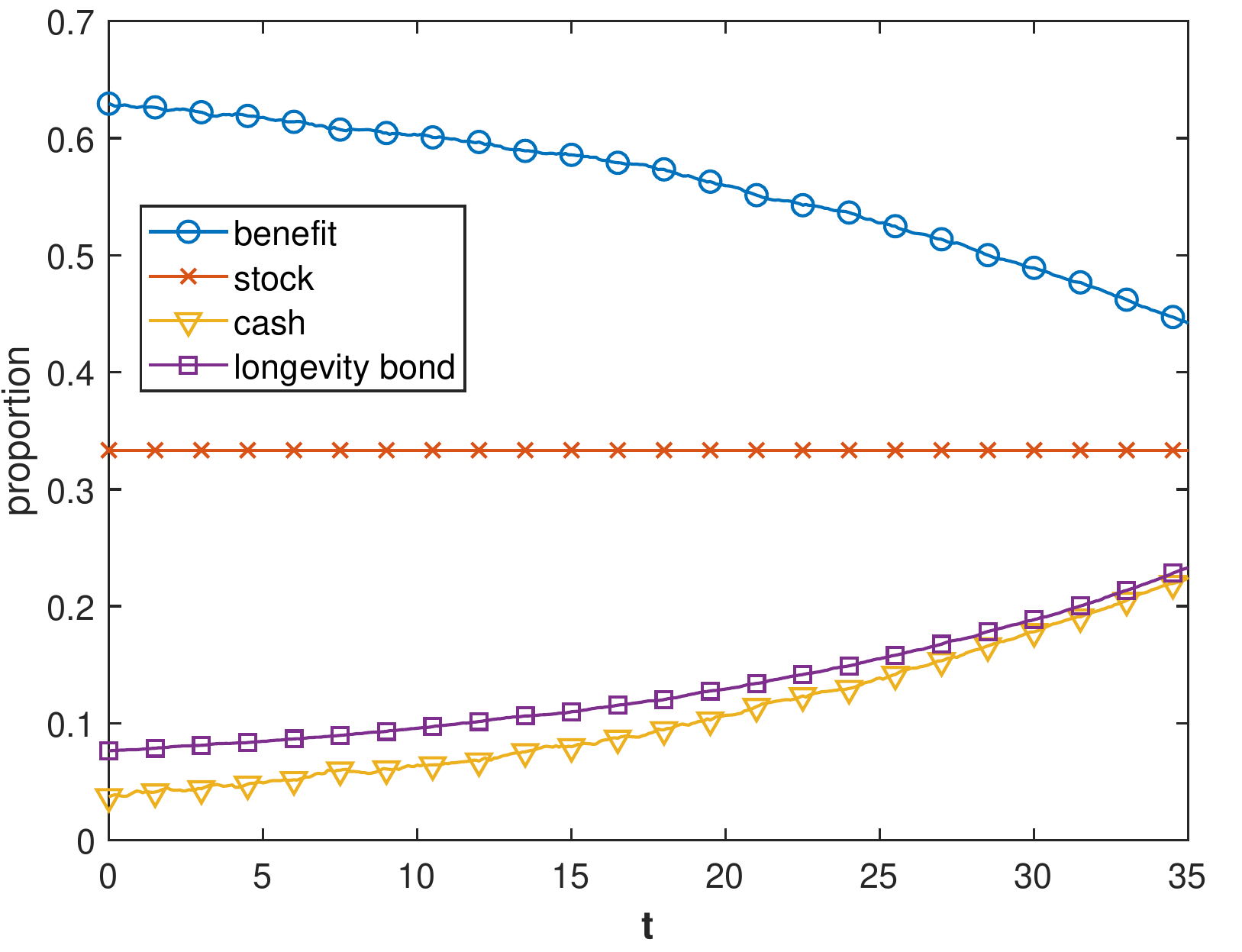}
    \caption{\textit{average optimal portfolio strategy and benefit withdrawal in sub-population case}}
    \label{fig:doublestrategy}
\end{center}
\end{figure}

\subsubsection{Comparison analysis}
In Section \ref{subsec:numerical1}, we show that in the single-population case, both the manager and members benefit from investing in the longevity bond. Now, we conduct a comparison study in the sub-population case to see whether the longevity bond still brings the advantage. For the simulation path 1, Figure \ref{fig:double} shows that investing in the longevity bond in general increases the members' benefit withdrawal and manager's compensation. For path 2, the members and manager benefit from investing in the longevity bond in the late 20-year time period, although the scheme suffers from some loss in early years. While the simulation path 3 shows that neither the manager nor members take advantage of longevity bond investment. To find out the reason for these phenomena, we check the simulated survival probabilities. We learn from the bottom plots in Figure \ref{fig:survival2} that on path 3, scheme members in Population 2 live much longer than anticipated. Whereas, the survival probability of longevity bond reference population (Population 1) on path 3 is lower. In this situation, the scheme suffers from severe longevity risk. , In this case, the longevity bond can not provide an efficient longevity risk hedge because the reference population displays shorter life expectancy. Since we assume the members are a sub-population of the longevity bond reference population, the members should have similar mortality trend with the reference population, although with slightly different behaviour. We believe that investing in the longevity bond can still be beneficial to both sides - manager and scheme members - in the two population case.
\begin{figure}[htbp]
\begin{center}
    \includegraphics[scale=0.48]{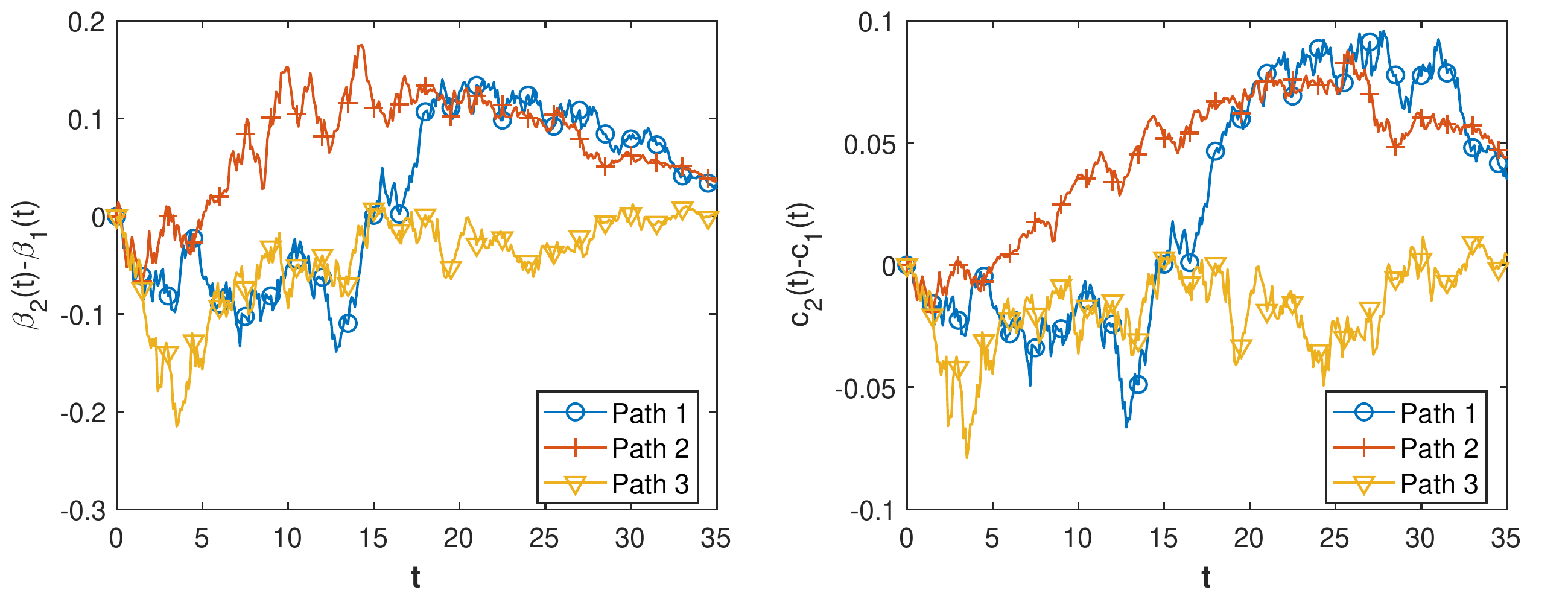}
    \caption{\textit{benefit and compensation improvement in sub-population case}}
    \label{fig:double}
\end{center}
\end{figure}

Overall, we claim that the longevity bond provides a way to hedge the longevity risk in the sub-population case. However, due to the presence of longevity basis risk, the hedging performance may be less effective than in the single-population case. The situation with multiple longevity bonds in the market may show different results and is worthy of an independent future study. The sensitivity analyses on the market price of longevity risk and the risk-sharing rule in the sub-population case deliver results which are similar to the single-population case.

\section{Conclusion}
\label{sec:conclusion}
This work studies the optimal portfolio strategy and benefit withdrawal rate for a pension scheme with an income-drawdown policy in the decumulation phase. The optimal solutions under single- and sub-population cases are obtained by applying the dynamic programming principle. The numerical study shows that the longevity bond can be used to hedge the longevity risk, and both members and manager benefit from the longevity bond investment. We believe that the development of a longevity market is required to provide a solution to capital markets for longevity risk hedging. Moreover, both members and manager benefit from an agreement on the risk-sharing rule in the long run. The problem with multiple longevity bonds issued on different populations is an interesting topic which we will explore in our future study.

\appendix
\section{Proof of Proposition \ref{prop:single}}
\label{ap:single}
\begin{proof}
The corresponding value function $v(Y,\lambda)$ of the optimisation problem \eqref{eq:problem} is given by
\eqstar{
v(Y,\lambda)={\mathop{\sup} \limits_{\alpha_S , \alpha_L ,\beta }}\ \Eb\Big[\int_0^\infty e^{-\int_0^s (r+\lambda_j(u))\dd u} \Big(\ln(\beta(s))+\phi \lambda_j(s)\ln( Y(s))\Big)\dd s \Big].
}
Applying Dynamic Programming Principle (DPP) to the value function gives
\eqlnostar{DPP}{
v(y,z)\ge\  &\Eb \left[\int_0^t e^{-\int_0^s(r +\lambda_j(u))\dd u}\Big[\ln(\beta(s))+\phi\lambda_j(s)\ln(Y(s))\Big]\dd s\right]+ \Eb\left[e^{-\int_0^t(r +\lambda_i(u))\dd u}v(Y_t^y,\lambda_t^\lambda)\right].
}
Applying It\^{o}'s formula to $e^{-\int_0^t(r +\lambda_i(u))\dd u}v(Y_t^y,\lambda_t^\lambda)$ and substituting it into the DPP \eqref{DPP} leads to
\eqstar{
0\geq &\Eb \left[\int_0^t e^{-\int_0^s(r +\lambda_j(u))\dd u}\Big[\ln(\beta(s))+\phi\lambda_j(s)\ln(Y(s))\Big]\dd s\right]-\Eb\left[\int_0^t(r +\lambda_j(s))e^{-\int_0^s(r +\lambda_j(u))\dd u}v(Y_s^y,\lambda_s^\lambda)\dd s\right]\\
&+\Eb \left[\int_0^te^{-\int_0^s(r +\lambda_j(u))\dd u}\mathcal{L}^{\alpha_S,\alpha_L,\beta}v(Y_s^y,Z_s^z)\dd s\right].
}
Dividing by $t$ and taking $t\rightarrow 0$ leads to 
\eqstar{
0 \ge \phi\lambda_j \ln y-(r+\lambda_j)v(y,\lambda)+\ln (\beta)+\mathcal{L}^{\alpha_S,\alpha_L,\beta}v(y,\lambda),
}
and we obtain the HJB equation
\eqstar{
0 = \phi\lambda_j \ln y-(r+\lambda_j)v(y,\lambda)+{\mathop{\sup} \limits_{\alpha_S , \alpha_L ,\beta }}\ \Big[\ln (\beta)+\mathcal{L}^{\alpha_S,\alpha_L,\beta}v(y,\lambda)\Big],
}
where 
\eqstar{
\mathcal{L}^{\alpha,b}v(y,\lambda)=&V_y\left[ry+\alpha_S\sigma_S\theta_S+\alpha_L\sum_{i,j=1}^{n-1}\sigma_L^{ij}\theta_j-\beta\right]+\mathscr{A}^\prime (t,\lambda)V_\lambda+\frac{1}{2}tr\left(\Sigma^\prime \Sigma V_{\lambda \lambda}\right)\\
&+\frac{1}{2}\left(\alpha_S^2\sigma_S^2+ \alpha_L^2\sum_{j=1}^{n-1}\left(\sum_{i=1}^{n-1}\sigma_L^{ij}\right)^2\right)V_{yy}+ \alpha_L\sum_{k=1}^{n-1}\left(\sum_{j=1}^{n-1}\left(\sum_{i=1}^{n-1}\sigma_L^{ij}\right)\sigma_{kj}\right)V_{y\lambda_k}.
}
%The first order conditions on $\beta$, $\alpha_S$ and $\alpha_L$ are:
%\eqstar{
%&0=\frac{1}{\beta}-V_y, \\
%&0=\theta_S V_y+\sigma_S\alpha_S V_{yy},\\
%&0=V_y\sum_{i,j=1}^{n-1}\sigma_L^{ij}\theta_j + V_{yy}\alpha_L\sum_{j=1}^{n-1}\left(\sum_{i=1}^{n-1}\sigma_L^{ij}\right)^2 +\sum_{k=1}^{n-1}\left(\sum_{j=1}^{n-1}\left(\sum_{i=1}^{n-1}\sigma_L^{ij}\right)\sigma_{kj}\right)V_{y\lambda_k}.
%}
%Solving the above first order conditions leads to the optimal solutions and HJB equation in Proposition \ref{prop:HJB}.

Solving the first order conditions on $\beta(t)$, $\alpha_S(t)$ and $\alpha_L(t)$ gives
\eqlnostar{eq:solu1}{
\beta^*(t) &= \frac{1}{V_y},\nonumber\\
\alpha_S^*(t) &= -\frac{\theta_S V_y}{\sigma_S V_{yy}},\nonumber\\
\alpha_L^*(t) &=-\frac{\sum_{i,j=1}^{n-1}\sigma_L^{ij}\theta_j V_y+\sum_{i,j=1}^{n-1}\sigma_L^{ij}\Sigma^\prime V_{y\lambda}}{\sum_{j=1}^{n-1}\left(\sum_{i=1}^{n-1}\sigma_L^{ij}\right)^2 V_{yy}}.}
Substituting $\eqref{eq:solu1}$ into the HJB equation leads to
\eqlnostar{eq:HJB}{
0=&\phi\lambda_j \ln y -(r +\lambda_j)V-\ln\left(V_y\right)+ryV_y-1+\mathscr{A}^\prime V_\lambda+\frac{1}{2}tr\left(\Sigma^\prime \Sigma V_{\lambda \lambda}\right)\nonumber\\
&-\frac{1}{2}\frac{V_y^2}{V_{yy}}\left[\theta_S^2-\frac{\left(\sum_{i,j=1}^{n-1}\sigma_L^{ij}\theta_j\right)^2}{\sum_{j=1}^{n-1}\left(\sum_{i=1}^{n-1}\sigma_L^{ij}\right)^2 }\right]-\frac{1}{2} \frac{\Big[ \sum_{k=1}^{n-1}\left(\sum_{j=1}^{n-1}\left(\sum_{i=1}^{n-1}\sigma_L^{ij}\right)\sigma_{kj}\right)V_{y\lambda_k}\Big]^2}{V_{yy}\sum_{j=1}^{n-1}\left(\sum_{i=1}^{n-1}\sigma_L^{ij}\right)^2}\nonumber\\
&-\frac{V_y}{V_{yy}}\frac{\left(\sum_{i,j=1}^{n-1}\sigma_L^{ij}\theta_j\right)\left(\sum_{k=1}^{n-1}\left(\sum_{j=1}^{n-1}\left(\sum_{i=1}^{n-1}\sigma_L^{ij}\right)\sigma_{kj}\right)V_{y\lambda_k}\right)}{\sum_{j=1}^{n-1}\left(\sum_{i=1}^{n-1}\sigma_L^{ij}\right)^2}.
}
We guess the solution to the PDE \eqref{eq:HJB} is of the following form
\eqstar{
V(y,\lambda)=G(\lambda)\ln y+H(\lambda),
}
with boundary conditions
\eqstar{
&{\lim\limits_{t \to \infty}}v(Y_t^y,Z_t^z)=0,& & {\lim\limits_{t \to \infty}}G(Z_t^z)=0, & &{\lim\limits_{t \to \infty}}H(Z_t^z)=0.
}
The PDE \eqref{eq:HJB} now becomes
\eqstar{
    0=&\phi\lambda_j\ln y -(r+\lambda_j)(G\ln y+H)-\ln G+\ln y+rG-1+\mathscr{A}^\prime (G_\lambda\ln y +H_\lambda)+\frac{1}{2}tr\left(\Sigma^\prime\Sigma\left(G_{\lambda \lambda}\ln y+H_{\lambda \lambda}\right)\right)\\
    &+\frac{1}{2}\left[\theta_S^2-\frac{\left(\sum_{i,j=1}^{n-1}\sigma_L^{ij}\theta_j\right)^2}{\sum_{j=1}^{n-1}\left(\sum_{i=1}^{n-1}\sigma_L^{ij}\right)^2 }\right]G+\frac{1}{2} \frac{\Big[ \sum_{k=1}^{n-1}\left(\sum_{j=1}^{n-1}\left(\sum_{i=1}^{n-1}\sigma_L^{ij}\right)\sigma_{kj}\right)G_{\lambda_k}\Big]^2}{G\sum_{j=1}^{n-1}\left(\sum_{i=1}^{n-1}\sigma_L^{ij}\right)^2}\nonumber\\
&+\frac{\left(\sum_{i,j=1}^{n-1}\sigma_L^{ij}\theta_j\right)\left(\sum_{k=1}^{n-1}\left(\sum_{j=1}^{n-1}\left(\sum_{i=1}^{n-1}\sigma_L^{ij}\right)\sigma_{kj}\right)G_{\lambda_k}\right)}{\sum_{j=1}^{n-1}\left(\sum_{i=1}^{n-1}\sigma_L^{ij}\right)^2}.
}
Separating the $\ln y$ terms and we get two ODEs
\eqstar{
0=&\phi\lambda_j-(r+\lambda_j)G+1+\mathscr{A}^\prime G_\lambda+\frac{1}{2}tr\left(\Sigma^\prime\Sigma G_{\lambda \lambda}\right),\\
0=&-(r+\lambda_j)H-\ln G+rG-1+\mathscr{A}^\prime H_\lambda+\frac{1}{2}tr\left(\Sigma^\prime\Sigma H_{\lambda \lambda}\right)+\frac{1}{2}\left[\theta_S^2-\frac{\left(\sum_{i,j=1}^{n-1}\sigma_L^{ij}\theta_j\right)^2}{\sum_{j=1}^{n-1}\left(\sum_{i=1}^{n-1}\sigma_L^{ij}\right)^2 }\right]G\\
&+\frac{\left(\sum_{i,j=1}^{n-1}\sigma_L^{ij}\theta_j\right)\left(\sum_{k=1}^{n-1}\left(\sum_{j=1}^{n-1}\left(\sum_{i=1}^{n-1}\sigma_L^{ij}\right)\sigma_{kj}\right)G_{\lambda_k}\right)}{\sum_{j=1}^{n-1}\left(\sum_{i=1}^{n-1}\sigma_L^{ij}\right)^2}+\frac{1}{2} \frac{\Big[ \sum_{k=1}^{n-1}\left(\sum_{j=1}^{n-1}\left(\sum_{i=1}^{n-1}\sigma_L^{ij}\right)\sigma_{kj}\right)G_{\lambda_k}\Big]^2}{G\sum_{j=1}^{n-1}\left(\sum_{i=1}^{n-1}\sigma_L^{ij}\right)^2}\nonumber.
}
We only need $G(\lambda)$ to get the optimal solutions so we solve the first ODE for $G(\lambda)$ and get
\eqstar{
G(\lambda)=\Eb\left[\int_0^\infty (\phi\lambda_j(s)+1)e^{-\int_0^s(r+\lambda_j(u))\dd u}\dd s\right].
}
$G(\lambda)$ has a dynamic version as
\eqstar{
G(t,\lambda)=\Eb_t\left[\int_t^\infty (\phi\lambda_j(s)+1)e^{-\int_t^s(r+\lambda_j(u))\dd u}\dd s\right].}
Substituting $G(t,\lambda)$ into \eqref{eq:solu1} gives the optimal solutions in Proposition \ref{prop:single}.
\end{proof}
\section{Proof of Proposition \ref{prop:singleOU}}
\label{ap:singleOU}
\begin{proof}
We first give the calculations for $A_0(t,s)$ and $A_1(t,s)$. Under OU setting \eqref{eq:singleOU}, denote by $h(t,s) = \Eb_t\left[e^{-\int_t^s \lambda(u) \dd u}\right]$ and applying It\^{o}'s formula to $e^{-\int_0^t \lambda(u)\dd u}h(t,s)$ and letting $\dd t$ term equals to $0$, we obtain 
\eqlnostar{eq:h}{
0 = -\lambda h+h_t+h_{\lambda}(a_1-b_1 \lambda)+\frac{1}{2}\sigma_1^2 h_{\lambda \lambda}
}
As the $\lambda(t)$ follows affine class model, we guess $h(t,s)$ has the following form
\eqstar{
h(t,s)=e^{A_0(t,s)-A_1(t,s)\lambda(t)},
}
with terminal conditions $h(s,s)=1$, $A_0(s,s)=0$ and $A_1(s,s)=0$.

Differentiating $h(t,s)$ and plugging into (\ref{eq:h}) leads to two ODEs 
\eqstar{
0 & = \frac{\del A_0}{\del t}-a_1A_1+\frac{1}{2}\sigma_1^2A_1^2,\\
0 & = -\frac{\del A_1}{\del t}+b_1 A_1-1.
}
Solving these ODEs gives $A_0(t,s)$ and $A_1(t,s)$ in Proposition \ref{prop:singleOU}. 

Similarly, denote by $h^\Qb(t,s) = \Eb^\Qb_t\left[e^{-\int_t^s \lambda(u) \dd u}\right]$ and assume $h^\Qb(t,s)=e^{A_0^\Qb(t,s)-A_1^\Qb(t,s)\lambda(t)}$, we obtain that $A_1^\Qb(t,s)=A_1(t,s)$. Moreover, we have $\sigma_L(t) = -A^\Qb_1(t,t+T)\sigma_1=-A_1(t,t+T)\sigma_1$.

Next, we show the derivation of $G(t,\lambda)$. According to Proposition \ref{prop:single}, we have
\eqstar{
G(t,\lambda)=&\phi\int_t^\infty \Eb_t\left[\lambda(s)e^{-\int_t^s \lambda(u)\dd u}\right]e^{-r(s-t)}\dd s+\int_t^\infty \Eb_t\left[e^{-\int_t^s \lambda(u)\dd u}\right]e^{-r(s-t)}\dd s.
}
To solve $\Eb_t\left[\lambda(s)e^{-\int_t^s \lambda(u)\dd u}\right]$, we denote by $\Zt(t)=\frac{\Ebt_t[e^{-\int_0^s \lambda(u)\dd u}]}{\Ebt[e^{-\int_0^s \lambda(u)\dd u}]}$, $m(t,s,\lambda)=e^{A_0(t,s)-A_1(t,s)\lambda(t)}$ and $D(t) = e^{-\int_0^t \lambda(u)\dd u}$ and have
\eqstar{
&\Eb_t\left[\lambda(s)e^{-\int_t^s \lambda(u)\dd u}\right]
=\Eb_t\left[\lambda(s)\frac{e^{-\int_t^s \lambda(u)\dd u}}{h(t,s)}\right]h(t,s)=\Eb_t\left[\lambda(s)\frac{\Zt(s)}{\Zt(t)}\right]h(t,s).
}
Let $\Eb_t\left[\lambda(s)\frac{\Zt(s)}{\Zt(t)}\right]=\Ebt_t[\lambda(s)]$, where $\Ebt[\cdot]$ denotes the expectation under measure $\Pbt$ which is equivalent to $\Pb$ and will be specify later. Applying It\^{o}'s formula to $\Zt(t)$ gives
\eqstar{
\dd \Zt(t) =& \dd D(t)\frac{m(t,s,\lambda)}{m(0,s,\lambda)}+\frac{D(t)}{m(0,s,\lambda)}\dd m(t,s,\lambda)\\
=&-\lambda(t)\Zt(t)\dd t+\frac{D(t)}{m(0,s,\lambda)}\left(m_t+m_{\lambda}(a_1-b_1 \lambda)+\frac{1}{2}\sigma_1^2 m_{\lambda \lambda}\right)\dd t+\frac{D(t)}{m(0,s,\lambda)}m_\lambda\sigma_1\dd W_1(t)\\
=&\sigma_1\frac{m_\lambda(t,s,\lambda)}{m(t,s,\lambda)}\Zt(t)\dd W_1(t)\\
=&-\sigma_1A_1(t,s)\Zt(t)\dd W_1(t)
}
Now, we can define $\Pbt$ through $\Zt(s)$
\eqstar{
\frac{\dd \Pbt}{\dd \Pb}=\Zt(s)=&\exp\left(-\int _0^s \sigma_1 A_1(u,s)\dd W_1(u)-\frac{1}{2}\int_0^s \sigma_1^2A_1^2(u,s)\dd u\right),
}
and we have
\eqstar{
\dd \Wt_1(t) &= \dd W_1(t)+\sigma_1 A_1(t,s)\dd t.
}

Thus, we have
\eqstar{
\dd \lambda(t)& = \left(a_1(t)-b_1\lambda(t)\right)\dd t+\sigma_1\left(\dd \Wt_1(t)-\sigma_1 A_1(t,s)\dd t\right) \\
&= \left(a_1(t)-b_1\lambda(t)-\sigma_1^2A_1(t,s)\right)\dd t+\sigma_1\dd \Wt_1(t),
}
taking expectation under $\Pbt$ leads to
\eqstar{
\frac{\dd \Ebt_t[\lambda(u)]}{\dd u }&=a_1(u)-\sigma_1^2A_1(u,s)-b_1\Ebt_t[\lambda(u)].
}
The solution of the above ODE is then obtained
\eqstar{
\Ebt_t[\lambda(s)]=\lambda(t)e^{-b_1(s-t)}+\int_t^s\left(a_1(u)-\sigma_1^2A_1(u,s)\right)e^{-b_1(s-u)} \dd u.
}
Substituting $\Ebt_t[\lambda(s)]$ gives the function $G(t,\lambda)$ in Proposition \ref{prop:singleOU}.
\end{proof}

\section{Proof of Proposition \ref{prop:doubleOU}}
\label{ap:doubleOU}
\begin{proof}
We first show the calculations for $C_0(t,s)$, $C_1(t,s)$ and $C_2(t,s)$. Denote by $f(t,s) = \Eb_t\left[e^{-\int_t^s \lambda_2(u)\dd u}\right]$ and applying It\^{o}'s formula to $e^{-\int_0^t \lambda_2(u)\dd u}f(t,s)$ and letting $\dd t$ term equals to $0$, we get
\eqlnostar{eq:f}{
0 = -\lambda_2f&+f_t+f_{\lambda_1
}(a_1-b_1\lambda_1)+f_{\lambda_2}(a_2-b_{21}\lambda_1-b_{22}\lambda_2)+\frac{1}{2}f_{\lambda_1 \lambda_1}\sigma_1^2\lambda_1\\ \nonumber
&+\frac{1}{2}f_{\lambda_2 \lambda_2}(\sigma_{21}^2\lambda_1+\sigma_{22}^2 \lambda_2)+f_{\lambda_1 \lambda_2}\sigma_1\sigma_{21}\lambda_1,
}
As $\lambda_1(t)$ and $\lambda_2(t)$ follow affine class models, we guess $f(t,s)$ has the following form
\eqstar{
f(t,s)=e^{C_0(t,s)-C_1(t,s)\lambda_1(t)-C_2(t,s)\lambda_2(t)},
}
with terminal conditions $f(s,s)=1$, $C_0(s,s)=0$, $C_1(s,s)=0$ and $C_2(s,s)=0$.

Differentiating $f(t,s)$ and plugging into (\ref{eq:f}) leads to three ODEs 
\eqstar{
0 & = -\frac{\del C_0}{\del t}-a_1C_1-a_2C_2+\frac{1}{2}\sigma_1^2C_1^2+\frac{1}{2}\left(\sigma_{21}^2+\sigma_{22}^2\right)C_2^2+\sigma_1\sigma_{21}C_1C_2,\\
0 & =-\frac{\del C_1}{\del t}+b_1C_1+b_{21}C_2,\\
0 & = -\frac{\del C_2}{\del t}+b_{22}C_2.
}
Solving these ODEs leads to $C_0(t,s)$, $C_1(t,s)$ and $C_2(t,s)$ in Proposition \ref{prop:doubleOU}.

Similar with Proposition \ref{prop:doubleOU}, to obtain $G(t,\lambda_1,\lambda_2)$, we first solve $\Ebt_t[\lambda_1(s)]$ and $\Ebt_t[\lambda_2(s)]$ under $\Pbt$ which is defined as
\eqstar{
\frac{\dd \Pbt}{\dd \Pb}=\Zt(s)=\exp\left(-\int _0^s \tilde{\theta}^\prime(u,s)\dd W(u)-\frac{1}{2}\int_0^s \left\|\tilde{\theta}(u,s)\right\|^2\dd u\right),
}
where
\eqstar{
\tilde{\theta}(u,s) = &\left[\begin{array}{ccc} 
\sigma_1 C_1(u,s)+\sigma_{21} C_2(u,s) \\
 \sigma_{22} C_2(u,s) \\
  0
\end{array}\right],\\
\dd &\Wt(u)=\dd W(u)+\tilde{\theta}(u,s) \dd u.
}
Then, solve the following two ODEs
\eqstar{
\frac{\dd \Ebt_t[\lambda_1(u)]}{\dd u }&=a_1(u)-\sigma_1^2C_1(u,s)-\sigma_1\sigma_{21}C_2(u,s)-b_1\Ebt_t[\lambda_1(u)],\\
\frac{\dd \Ebt_t[\lambda_2(u)]}{\dd u }&=a_2(u)-\sigma_1\sigma_{21}C_1(u,s)-(\sigma_{21}^2+\sigma_{22}^2)C_2(u,s) -b_{21}\Ebt_t[\lambda_1(u)]-b_{22}\Ebt_t[\lambda_2(u)],
}
we obtain
\eqstar{
\Ebt_t[\lambda_1(s)]=&\lambda_1(t)e^{-b_1(s-t)}+\int_t^s\left(a_1(u)-\sigma_1^2C_1(u,s)-\sigma_1\sigma_{21}C_2(u,s)\right)e^{-b_1(s-u)} \dd u,\\
\Ebt_t[\lambda_2(s)]=&\frac{b_{21}}{b_1-b_{22}}\lambda_1(t)\left[e^{-b_1(s-t)}-e^{-b_{22}(s-t)}\right]+\lambda_2(t)e^{-b_{22}(s-t)}\\
&+\frac{b_{21}}{b_1-b_{22}}\int_t^se^{-b_1(s-u)}\left[a_1(u)-\sigma_1^2C_1(u,s)-\sigma_1\sigma_{21}C_2(u,s)\right]\dd u\\
&-\int_t^s e^{-b_{22}(s-u)}\Big[\frac{b_{21}}{b_1-b_{22}}a_1(u)-a_2(u)+\sigma_1\sigma_{21}C_1(u,s)\\
&+\left(\sigma_{21}^2+\sigma_{22}^2-\frac{b_{21}\sigma_1\sigma_{21}}{b_1-b_{22}}\right)C_2(u,s)\Big]\dd u.
}
\end{proof}
\bibliographystyle{abbrvnat}
\bibliography{bibfile}
\nocite{*}
\end{document}